\documentclass[a4paper, 11pt]{article}
\usepackage{amsmath,amsthm,amssymb}
\usepackage{graphicx}
\usepackage{color}
\usepackage{float}
\usepackage{dsfont}
\usepackage{epstopdf}
\usepackage{longtable}
\usepackage{algorithm}
\usepackage{algpseudocode}
\usepackage[hang]{subfigure}
\usepackage{natbib}
\usepackage[title]{appendix}
\usepackage{ulem}
\usepackage{color}
\usepackage{textcomp}
\usepackage{algorithm}
\usepackage{algorithmicx}
\usepackage{lscape}
\newtheorem{thm}{Theorem}
\newtheorem{coro}[thm]{Corollary}

\newtheorem{prop}[thm]{Proposition}
\newtheorem{assumption}{Assumption}
\newtheorem{assumptionstar}{Assumption}

\newcounter{rmk}
\newcommand\rmk[1]{\vspace*{1mm} \par \stepcounter{rmk}{\noindent \bf Remark \thermk}. {#1}\vspace*{1mm}}

\usepackage[a4paper]{geometry}
\definecolor{red}{RGB}{139,0,18}
\definecolor{lightred}{RGB}{186,25,31}
\definecolor{blue}{RGB}{0,56,108}
\definecolor{lightblue}{RGB}{69,100,139}
\renewcommand\emph[1]{{\color{red}\itshape #1}}

\def\t{\mathrm{\scriptscriptstyle T} }
\def\tr{\text{tr}}
\def\sign{\text{sign}}

\def\mR{\mathds{R}}

\def\t{\mathrm{\scriptscriptstyle T} }
\def\argmin{\mathop{\arg\min}}
\def\argmax{\mathop{\arg\max}}

\title{Non-zero block selector: A linear correlation coefficient measure for blocking-selection models}
\author{Weixiong Liang$^\text{a}$ and Yuehan Yang$^\text{b}$
\thanks{Corresponding Author. Email: yyh@cufe.edu.cn.}\\
{\it \small $^\text{a}$Department of Statistics and Data Science, Southern University of Science and Technology,}\\
{\it \small Shenzhen, China}\\
{\it \small $^\text{b}$School of Statistics and Mathematics, Central University of Finance and Economics,}\\
{\it \small Beijing, China}}
\date{}

\begin{document}
\maketitle
\begin{abstract}
Analyzing multiple-group data is crucial in fields such as genomics, finance, and social sciences. This study investigates a block structure that consists of covariate and response groups. It examines the block-selection problem of high-dimensional models with group structures for both responses and covariates, where both the number of blocks and the dimension within each block are allowed to grow larger than the sample size. A novel framework is proposed, comprising a block-selection model and a non-zero block selector (NBS) for identifying and utilizing these structures. We establish the uniform consistency of the NBS and introduce three estimators based on the NBS to enhance modeling efficiency. Theoretical results show that these estimators achieve the oracle property, are consistent, jointly asymptotically normal, and efficient in modeling ultra-high-dimensional data. Extensive simulations under complex data settings demonstrate the superior performance of the proposed method, while an application to gene-expression data further validates its practical utility.
\end{abstract}


\section{Introduction}
In the context of extensive data and multiple data groups, the identification of relationships between groups and the screening process would deviate from traditional modeling techniques. For example, electronic health data typically contains a massive amount of information received from patients. Symptoms of a disease usually appear in groups and are related to a part of abnormal inspection results. These structures naturally exist in many fields and can be generally determined by relevant facts.

In practical applications such as those involving genomic data, where the dimensions of covariates and responses are typically substantial, block structures always exist, such as the gene-group effects \citep{hummel2008globalancova}, the association between gene groups \citep{park2008penalized, zhang2010new}, and the protein and DNA association detection \citep{zamdborg2009discovery}. To illustrate the practical importance of block selection in the high-dimensional, multi-response, linear model with block structures \eqref{eq model}, we present the following two examples.

\begin{itemize}
\item[] \textbf{Example 1. International stock market}
\item[] International stock market studies often involve several major stock markets globally \citep{liu2022international}. Let $(Y_1,\dots, Y_J)$ represent stock market indices from various countries, such as the SP500 index for the American stock market, and $(X_1,\dots, X_K)$ represent the information of firms, such as the financial figures, from around the world. Indexes (responses) and information of firms (covariates) from the same country or same industry tend to be highly correlated and can be grouped within a single block. Researchers have explored how financial contagion occurs within each block and between different blocks \citep{akhtaruzzaman2021financial,guo2021tail}. In recent studies, there has been a growing interest in understanding how global events, such as the COVID-19 pandemic, impact the interconnectedness network of stock market volatility, with a particular focus on grouping effects \citep{cheng2022impact}. The model \eqref{eq model} accurately captures these practical phenomena.

\item[] \textbf{Example 2. Gene expression and fMRI intensity}
\item[] Genomic data, such as the pathway group structures for gene expressions and brain functional regions for fMRI intensity responses, often exhibit high-dimensional and block structures \citep{zamdborg2009discovery,stein2010voxelwise}. For example, single nucleotide polymorphisms (SNPs) are often grouped into genes, and genes are grouped into biological pathways. Block selection is suitable for biological studies aimed at detecting models in which each trait and covariate are embedded within biological functional groups, such as genes, pathways, or brain regions \citep{li2015multivariate}. In recent years, the analysis of associations between high-dimensional covariates and responses with natural group (block) structures has also been a subject of study in various genomic-data analyses \citep{park2008penalized,zhang2010new}.
\end{itemize}

When performing modeling and prediction with such data as responses and covariates, it is often beneficial to identify the corresponding groups of responses and covariates, which can help improve predictive performance and enhance the interpretation of the fit. To simultaneously describe data with multiple covariate and response groups, we consider models with block structures. That is, we consider the following high-dimensional model with a block structure with respect to both response and covariate groups:
\begin{equation}\label{eq model}
\big( Y_1, \dots, Y_J \big)
= \big( X_1, \dots, X_K \big)
\left( \begin{array}{ccc}
B_{11} & \cdots & B_{1J}\\
\vdots & \ddots & \vdots\\
B_{K1} & \cdots & B_{KJ}
\end{array}\right)
+ \big(E_1, \dots, E_J \big),
\end{equation}
where $Y_j$ and $X_k$ denote the $j$th response group and $k$th covariate group, respectively, and $B_{kj}$ denotes the coefficient matrix of $X_k$ to $Y_j$ such that $B_{kj} = \{\beta^{kj}_{ll'}\}$ and $E_j$ denotes the random error matrix. We denote $\Sigma_j$ the covariance matrix of $E_j$ and assume that it is non-negative definite with positive diagonal elements and the diagonal elements are bounded. $X_k$'s are given and the related covariates will be grouped into the same group. We do not require that different covariate groups are independent.
In this model, each pair $(k,j)$, for $k=1,\dots K$ and $j = 1,\dots,J$ represents a block. We are interested in which blocks are relevant to the model, i.e., the $(k,j)$th block is relevant when at least one element in the block is relevant,  thus defining
$$J_1 = \{(k,j):  \exists~ \beta^{kj}_{ll'} \neq 0 \}.$$
We allow the three dimensions, i.e., the dimensions of blocks, covariates, and responses, to be larger than the number of observations, and assume $|J_1| \ll K \times J$ when the latter becomes large. When $K \times J = 1$, the above model reverses to the regular multi-response, high-dimensional regression model. To fit this model, we should screen the relevant blocks and identify the covariates. We define the screening of relevant blocks in the model as ``block selection''. The blocking effect is similar to but different from the grouping effect because we study the ``blocked variables''. The blocking effect takes into account not only the grouping effect within the covariates but also the grouping effects on both the responses and the covariates. After organizing both the covariates and the responses into groups, the relationships between these grouped covariates that contribute to the grouped responses are depicted as blocks, resulting in the blocking effect.
Unlike the grouping effect, the blocking effect has rarely been studied. In this article, we develop a block selection method for high-dimensional models with group structures for both responses and covariates.

\subsection{Related work}
Early work studying the feature selection problem mainly focuses on group-structured models or multi-response models. Based on the sparsity assumptions, multi-response models require estimating a sparse coefficient matrix, which can be achieved by flexible application of penalized regularizations, such as the lasso \citep{tibshirani1996regression}, adaptive lasso penalty \citep{zou2006adaptive}, elastic net \citep{zou2005regularization}, MCP \citep{zhang2010nearly}, and group lasso \citep{yuan2006model}. Other strategies include sequential approaches such as forward regression in high-dimensional space \citep{wang2009forward}, step regression in high-dimensional and sparse space \citep{ing2011stepwise}, orthogonal matching pursuit\citep{cai2011orthogonal}, and sequential lasso \citep{luo2014sequential}. Furthermore, to obtain an accurate multi-response model, researchers have proposed many targeted methods, such studies including the dimension reduction technique \citep{yuan2007dimension, chen2012sparse}, the extension of lasso \citep{turlach2005simultaneous}, the use of $\ell_{1} / \ell_{2}$ regularization \citep{obozinski2011support}, the investigation of a regularized multivariate regression for identifying master covariates \citep{peng2010regularized}.

Similar to block structures, group structures present many challenges, mainly due to the difficulty in exploring the information within them. Although the above penalized regularizations would perform well in various situations, they may not be able to recover a model with block or group structures because ignoring these structures can result in insufficiency and hurt model inference \citep{dudoit2003multiple, stein2010voxelwise}, particularly when dealing with a large number of covariates and responses.

Two recent studies have successfully addressed the multiple-response model with group structures. First is the multivariate sparse group lasso (MSGLasso) \citep{li2015multivariate}, which proposed a mixed coordinate descent algorithm that utilized group-structure information to determine the descent direction. The second is the Sequential Canonical Correlation Search (SCCS) \citep{luo2020feature}, which considered group-structure information in the canonical correlation to identify non-zero blocks and used the Extended Bayesian Information Criterion (EBIC) for feature selection. However, there is a drawback worth mentioning for both methods: as the dimensions of both covariates and responses increase, they inevitably require excessive computational resources.

Most relevant to our study, \citet{hu2023response} introduced a best-subset selector for response selection. Specifically, they introduced an indicator for the response and proposed a response best-subset selection model $Y\Delta = X\Theta + \mathcal{E} \Delta$. We extend their work to block structures and construct a selector for each coefficient block, rather than solely focusing on the response. The most significant contribution of this extension is that the proposed indicator is well-suited for block selection problems and more extensive data sets. Additionally, the selector for blocks comprising both the response and covariate groups is considerably more intricate than the response selector. This complexity extends to the model itself, with the associated theoretical guarantees and computations still being lacking. Furthermore, the best-subset selector introduced by \citet{hu2023response} is only suitable for low dimensions. Our extension allows it to be applied to larger data sets, where the number of groups and the numbers of responses and covariates within each group can exceed the sample size.

\subsection{Contributions and organization of the paper}
Given the practical need for block selection and the existing research gap in this area, this paper proposes a block-selection model and introduces a novel method called the non-zero block selector (NBS) to address the problem of non-zero block selection. Specifically, we offer three stable and straightforward measures for selecting both blocks and coefficients. The main contributions of this paper are as follows:

First, we propose a new linear correlation coefficient measure for block-structured data and establish its properties. Based on this measure, we introduce an indicator, $\Delta_{kj}$, for each coefficient block, $B_{kj}$, and propose an efficient algorithm to calculate the indicator. This measure allows us to describe the correlation between the covariate group, $X_k$, and the response group, $Y_j$, and directly select the relevant blocks from the data. We provide the uniqueness of the estimator and establish its related asymptotic property.

Second, we provide three effective penalized functions and algorithms for the block-selection model. We discuss various scenarios of the block-selection model, including situations where the covariates have full rank and high-dimensional settings. We also explore whether there is a requirement for sparsity in these situations. For each estimator, we provide uniqueness and asymptotic properties under regular conditions. Most relevant to our work, \cite{luo2020feature} considers similar modeling without providing theoretical guarantees regarding the estimation accuracy. We fill this gap by proving that the proposed estimates achieve the optimal minimax rate for the upper bound of the $l_2$ norm error. In comparison with MSGlasso \citep{li2015multivariate}, we significantly expand the theoretical constraints regarding data dimension without compromising accuracy.

One significant novelty of this NBS selector is its computational efficiency. Dealing with coefficient blocks using traditional methods often involves iterating through all the elements of each coefficient block, such as the SCCS, which can be computationally expensive. The NBS significantly reduces computational costs by identifying irrelevant blocks through direct calculations. We prove that the operator can select each relevant block with a high probability, leveraging block-structure information to accelerate the algorithm and enhance estimation performance.

To evaluate the performance of the NBS, we conduct extensive simulation studies and real data analysis. The results demonstrate that the NBS outperforms existing methods in terms of prediction accuracy and model interpretability. Overall, the NBS provides a useful tool for modeling high-dimensional, multi-response data with block structures, which has wide-ranging applications in various fields, including genomics, finance, and the social sciences.

The remainder of the article is organized as follows. In Section 2, we introduce the model and the NBS. Section 3 provides three effective penalized functions and algorithms for the block-selection model and the theoretical properties. We report our simulation studies for the comparison of the NBS and other methods in Section 4. Section 5 provides an analysis of a real data set. Section 6 concludes the article with a discussion. Appendix A provides additional simulation results, and the proofs are presented in Appendix B.

\section{Model and Indicator}
In this section, we introduce the proposed block-selection model and its associated selection indicator. Subsection 2.1 presents the block model and its integration with the selection indicator. Subsection 2.2 introduces the indicator optimization function and establish some properties of it. Subsection 2.3 discuss the estimation of the indicator optimization function. Subsection 2.4 gives a practical approach to obtain the tuning parameter and provides a straightforward algorithm for obtaining the selector.

We first introduce some notation and framework. Consider a response variable set, $\mathcal{Y} = (\mathcal{Y}_1,\dots, \mathcal{Y}_J)$, and a covariate set, $\mathcal{X} = (\mathcal{X}_1,\dots,\mathcal{X}_K)$, where $\mathcal{Y}_j$ and $\mathcal{X}_k$ represent the $j$th response group and the $k$th covariate group, respectively. Each pair of these groups forms a block, containing $|\mathcal{X}_k| = p_k$ by $|\mathcal{Y}_j| = q_j$ elements, where $\sum_{k=1}^K p_k = P$ and $\sum_{j=1}^J q_j = Q$. The numbers of groups, $K$ and $J$, can exceed the sample size $n$, making both $P$ and $Q$ larger than $n$. Additionally, within each block $(k, j)$, $p_k$ may also exceed $n$. We use $Y_{n \times Q} = (Y_1,\dots, Y_J)$ and $X_{n \times P} = (X_1,\dots,X_K)$ to represent the observations of the $\mathcal{Y}$ and $\mathcal{X}$, respectively. For ease of analysis, we assume that each column of the response matrix, $Y$, and covariate matrix, $X$, are standardized. For the multivariate modeling in \eqref{eq model}, we set $B = \{B_{kj}\}_{K \times J}$, where each $B_{kj}$ denotes the coefficient matrix of $X_k$ and $Y_j$.

\subsection{Block-selection model}
In this part, we introduce the block-selection model and the indicator for coefficient block selection. Specifically, we assume that the model is sparse, meaning that it contains many irrelevant blocks. In other words, all the covariates are irrelevant to the responses in this block. To identify irrelevant blocks and detect relevant ones, we introduce the selection indicator, $\Delta = (\Delta_{kj})_{K \times J}$, where
\begin{equation}\label{eq deltakj original}
\Delta_{kj} = \left\{\begin{array}{ll}
1 & \  \exists~ \beta^{kj}_{ll'} \in B_{kj}, \beta^{kj}_{ll'} \neq 0 \\
0 & \text{otherwise}.
\end{array}\right.
\end{equation}
The above indicator represents whether the block $(k,j)$ is relevant or irrelevant to the model and we have $B_{kj} = B_{kj} \bullet \Delta_{kj}$. For instance, if $\Delta_{23}=1$, it means that $B_{23}$ contains non-zero elements and $\Delta_{23} = 0$ indicates that $B_{23} = 0$. Incorporating the indicator matrix into the model helps us identify the relevant coefficient blocks. Thus, the high-dimensional response and covariate model, \eqref{eq model}, is equivalent to the following block-selection model:
\begin{equation}\label{eq model0}
\big( Y_1, \dots, Y_J \big)
= \big( X_1, \dots, X_K \big)
\left( \begin{array}{ccc}
B_{11} \bullet \Delta_{11} & \dots, & B_{1J} \bullet  \Delta_{1J} \\
\vdots & \ddots & \vdots\\
B_{K1}\bullet  \Delta_{K1} & \dots, & B_{KJ}\bullet  \Delta_{KJ}
\end{array}\right)
+ \big( E_1, \dots, E_J \big),
\end{equation}
where $\bullet$ indicates scalar multiplication. The block-selection model presented above combines the indicator and coefficient matrix for each block. The block-selection model simplifies the original model, \eqref{eq model}, by introducing a selection indicator. Our goal is to estimate this selection indicator, thereby reducing the complexity associated with estimating the coefficient matrix of the model.

We also introduce a concise version of the block-selection model, that is, set
\begin{equation*}
B_\Delta = \left( \begin{array}{ccc}
B_{11} \bullet \Delta_{11} & \dots & B_{1J} \bullet  \Delta_{1J} \\
\vdots & \ddots & \vdots\\
B_{K1}\bullet  \Delta_{K1} & \dots & B_{KJ}\bullet  \Delta_{KJ}
\end{array}\right).
\end{equation*}
Then, \eqref{eq model0} is equivalent to the following equation:
\begin{equation}\label{eq selection model}
Y = X B_{\Delta} + E,
\end{equation}
and for each $j$, we have
\[ Y_j = X B_{\Delta j} + E_j,\]
where $B_{\Delta j} = (B_{1j}\bullet\Delta_{1j}\dots B_{Kj}\bullet\Delta_{Kj})^\t$ relating to $B_j = (B_{1j} \dots B_{Kj})^\t$. The selection model, \eqref{eq selection model}, can be regarded as an extension of the traditional model $Y = X B + E$. Notably, conventional penalized regularization methods face challenges when applied to the model above while maintaining computational efficiency and estimation accuracy. In the following, we delve into the methodology for solving the indicator within the context of the block-selection model.

\subsection{Indicator optimization function}\label{section 2.2}
In this section, we focus on calculating the block-selection indicator by introducing an indicator optimization function. As mentioned earlier, we defined an indicator matrix $\Delta = \{\Delta_{kj}\}_{K \times J}$. According to the block-selection model, determining whether a block, $B_{kj}$, equals 1 or 0 is equivalent to deciding whether the covariate group, $X_k$, is relevant to the response group, $Y_j$, in this model. Based on this definition, we establish the following equivalence:
\begin{equation}\label{eq equivalence}
\Delta_{kj} = 1 ~\rightleftarrows~ B_{kj} \neq 0 ~\rightleftarrows~ \|X_k B_{kj}\|^2_F > 0,
\end{equation}
where $\|\cdot \|_F$ denotes the Frobenius norm. To calculate the value of indicator $\Delta_{kj}$, we consider constructing a correlation function for each pair of covariate and response groups. Recalling that $|\mathcal{X}_k| = p_k$, we first calculate the following oracle least squares estimator for each block $(k,j)$, when $p_k \leqslant n$:
\begin{equation}\label{eq esti 1}
\hat B_{kj} = (X_k^TX_k)^{-1}X_k^T Y_j.
\end{equation}
Since we allow the dimension within the block larger than the number of observations, when $p_k > n$, the inverse of $(X_k^\t X_k)$ may not be unique. We solve this problem by adding the sparse requirement. That is, we consider the inverse of subset $(X^\t_{\hat S_{kj}}X_{\hat S_{kj}})$ where $|\hat S_{kj}| < n$ and the subset $\hat S_{kj}$ can be obtained from the lasso \citep{tibshirani1996regression}, sure independent screening (SIS) \citep{fan2008sure}, etc. That is, when $p_k > n$, we obtain the following estimator for block $(k,j)$:
\begin{equation}\label{eq esti 2}
\hat B_{kj} = (X_{\hat S_{kj}}^\t X_{\hat S_{kj}})^{-1}X_{\hat S_{kj}}^\t  Y_j.
\end{equation}
For notational simplicity, we include the discussion about the difference between the above two estimators, \eqref{eq esti 1} and \eqref{eq esti 2}, in the next section (combined with the joint estimation with or without penalized regularizations) and omit their differences here. In the following, we adhere to the notation and situation in \eqref{eq esti 1} and $p_k \leqslant n$.

We consider proposing an indicator optimization function that calculates the indicator by balancing the values between $\Delta_{kj}$ and the dual transformation, $(1 - \Delta_{kj})$. Based on \eqref{eq equivalence}, equivalently, the optimization function can be obtained by balancing the values between $\|Y_j - X_k \hat B_{kj} \|^2_F $ and $\|X_k\hat B_{kj}\|^2_F$. We also consider the degrees of freedom for both values, which are inspired by the adjusted R squared, as follows:
\[ n-p_k-1 ~\text{and}~ p_k-1,\]
where the former is the degree of $\|Y_{j}- X_k \hat B_{kj}\|_{F}^{2}$ and the latter is the degree of $\|X_k\hat B_{kj}\|_{F}^{2}$.
Combining the degrees, we introduce the following correlation function to the indicator, that is, for the block $(k,j)$:
\begin{equation}\label{eq deltakj}
\hat \Delta_{kj} = \argmin_{\Delta_{kj} \in \mathcal{V}}
\Big(\frac{\left\|Y_{j}-P_{X_{k}} Y_{j}\right\|_{F}^{2}}{n-p_k-1} \Delta_{k j}
+\gamma \frac{\left\|P_{X_{k}} Y_{j}\right\|_{F}^{2}}{p_k -1}\left(1-\Delta_{k j}\right)
\Big),
\end{equation}
where $\mathcal{V} = \{0,1\}$, $\gamma $ is a tuning parameter that will be discussed later and $P_{X_k} = X_k(X_k^\t X_k)^{-1}X_k^\t $ denotes the projection matrix of $X_k$.

To establish properties of the above optimization, we impose the following assumptions for model \eqref{eq model}.

\begin{assumption}{(Low-dimensional setting requirement)}\label{assumption p}
If $p_k <n$, we require $\max_{k}p_k \le c_1 n^\xi$, where $0<\xi<1/2$, and the number of covariate groups satisfies $K \leqslant c_1' n^{\xi/2}$ for some constants $c_1,c_1' \ge 0$.
\end{assumption}
\begin{assumptionstar}{(Sparse requirement)}\label{assumption p'}
If the condition in Assumption \ref{assumption p} does not hold, let $S_{kj} = \{i:\|B_{kj,i \cdot }\|_F \ne 0 \}$, where $B_{kj,i \cdot }$ denotes the $i$th row of $B_{kj}$. Then, $\max_{(k,j)}|S_{kj}| \le c_1 n^\xi$, with $0<\xi<1/2$, and the number of covariate groups satisfies $K \leqslant c_1' n^{\xi/2}$ for some constants $c_1,c_1' \ge 0$.
\end{assumptionstar}
\begin{assumption}\label{assumption Ej}
For any pair $(k,j)$, the random errors $E_j$ are independent of $X_k$. The covariance matrix $\Sigma_j$ of $E_j$ is assumed to be non-negative definite with positive diagonal elements, and these diagonal elements are bounded.
\end{assumption}
\begin{assumption}\label{assumption correlation}
Define $T_{kj} = X_kB_{kj}$ and let $\langle A,B \rangle = \tr(A^\t B)$. For all $k = 1,2,...,K$, there exist positive constaints $c_2,c_3, \eta, \tau  > 0$, such that $\eta + \xi + \tau \le 1$ and the following hold,
\begin{align*}
&\max_{k' \ne k} \lambda_{max}(X_{k'}^\t P_{X_k}X_{k'}) \le c_2 n^{\eta/2}, \\
&\max_{k' \ne k} \langle T_{kj},T_{k'j} \rangle \le c_3 n^{\tau}.
\end{align*}
\end{assumption}
\begin{assumption}\label{assumption sparsity}
The singular values of $B_{kj}$ are bounded, i.e., for some constant $c_4 >0$,
$$\max_{(k,j)} \sum_i |\sigma_i(B_{kj})| \le c_4n^{\eta/2},$$
where $\sigma(\cdot)$ denote the singular value of some matrix.
\end{assumption}
Assumption \ref{assumption p} establishes a low dimensional constraint for cases where the data satisfies this assumption. Specifically, the number of the covariates in each block, $p_k$, is allowed diverge as $n$ grows, at a rate governed by $\xi$, and the number of covariate groups $K$ can also diverge, with at a slower rate of $\xi/2$. This assumption on $p_k$ and $K$ are relatively mild under low dimensional settings. Assumption \ref{assumption p'} addresses situations where Assumption \ref{assumption p} does not hold. In such cases, it allows some $p_k$ values to exceed $n$, but imposes restrictions on the number of relevant covariates, requiring them to satisfies the same requirements in Assumption \ref{assumption p}. This ensures that the estimator \eqref{eq esti 2} retains properties similar to those of \eqref{eq esti 1}.

Assumption \ref{assumption Ej} is a common assumption that bounds the variability of the error terms $E_j$, and it can derive that the eigenvalues of $E_jE_j^\t$ are bounded. Assumption \ref{assumption correlation} imposes a constraint of the correlation between $X_k$ and $X_{k'}$, which rules out scenarios where covariates from different groups exhibit strong correlations. This assumption is similar to some commonly used conditions that address the issue of strong collinearity, such as $\lambda_{max}(\Sigma) \le c_4n^\tau$\citep{fan2008sure}, $\lambda_{max}(\Sigma)/\lambda_{min}(\Sigma) \le c_4n^\tau$\citep{wang2016high}, among others. Assumption \ref{assumption sparsity} establishes a natural requirement for $B_{kj}$ in low dimensional settings and serves as a reasonable sparsity requirement for $B_{kj}$ in high-dimensional settings.

Based on the definition of $T_{kj}$ in Assumption \ref{assumption correlation}, we have
\begin{align*}
E\|X_k \hat B_{kj}\|^2_F &=  E\|P_{X_k}Y_j\|^2_F \nonumber\\
&=\langle T_{kj},T_{kj} \rangle
+ E\langle P_{X_k}E_j,P_{X_k}E_j \rangle
+ \sum\limits_{k',k'' \ne k} \langle P_{X_k}T_{k''j},P_{X_k}T_{k'j} \rangle
+ 2\sum\limits_{k' \ne k} \langle T_{kj},T_{k'j} \rangle.
\end{align*}
When $\Delta_{kj} = 1$, implying $T_{kj} \ne 0$, the terms are bounded as follows:
\begin{align*}
& (1)~ E\langle P_{X_k}E_j,P_{X_k}E_j \rangle = E\sum\limits_{i = 1}^{n} \lambda_i(P_{X_k}E_jE_j^\t)\\
&~~~~~~~~~~~~~~~~~~~~~~~~~~~~\le E\sum\limits_{i = 1}^{n} \lambda_i(P_{X_k})\lambda_i(E_jE_j^\t)
\le E\tr(P_{X_k}) \cdot \max_{i} \lambda_i(E_jE_j^\t) = O(n^\xi), \\
& (2)~ \sum\limits_{k',k'' \ne k} \langle P_{X_k}T_{k''j},P_{X_k}T_{k'j} \rangle
\le (K-1)^2 \cdot \max_{k' \ne k} \langle P_{X_k}T_{k'j},P_{X_k}T_{k'j} \rangle \\
&~~~~~~~~~~~~~~~~~~~~~~~~~~~~~~~~~~~~~~~\le (K-1)^2 \cdot \max_{k' \ne k} \Big[ \sum\limits_{i} \lambda_i(X_{k'}^\t P_{X_k} X_{k'}) \lambda_i(B_{k'j}B_{k'j}^\t) \Big] \\
&~~~~~~~~~~~~~~~~~~~~~~~~~~~~~~~~~~~~~~~ \le (K-1)^2 \cdot c_2 c_4 n^\eta = O(n^{\eta + \xi}),\\
& (3)~ 2\sum\limits_{k' \ne k} \langle T_{kj},T_{k'j} \rangle \le 2(K-1)\max_{k' \ne k} \langle T_{kj},T_{k'j} \rangle
= O(n^{\tau+\xi/2}).
\end{align*}
The last equality of $(1)$ follows directly from Assumptions \ref{assumption p} and \ref{assumption Ej}. Specifically, $\tr(P_{X_k}) \le c_1 n^\xi$ and $\max_{i} \lambda_i(E_jE_j^\t)$ is bounded. For $(2)$ and $(3)$, their respective orders are derived using Assumption \ref{assumption p}, \ref{assumption correlation} and \ref{assumption sparsity}. Thus, the overall order is
$$E\|P_{X_k}Y_j\|^2_F = O(n) + O(n^{\xi}) + O(n^{\eta+\xi}) + O(n^{\tau + \xi/2 }) = O(n).$$
Similarly, when $\Delta_{kj} = 0$, implying $T_{kj} = 0$, we have
$$E\|P_{X_k}Y_j\|^2_F = 0 + O(n^\xi) + O(n^{\eta+\xi}) + 0 = O(n^{\eta+\xi}).$$
The expected squared Frobenius norm of the error term is given by
\begin{align*}
E\|Y_j - X_k \hat B_{kj}\|^2_F
&= E(\|Y_j - P_{X_k} Y_j \|_F^2)  \nonumber \\
&= \sum\limits_{k'\ne k} \langle T_{k'j},T_{k'j} \rangle
+ \sum\limits_{\substack{k'\ne k'' \\ k',k'' \ne k}} \langle T_{k'j},T_{k''j} \rangle
- \sum\limits_{k',k'' \ne k} \langle P_{X_k}T_{k''j},P_{X_k}T_{k'j} \rangle \\
&+ E\langle (I_n-P_{X_k})E_j,(I_n-P_{X_k})E_j \rangle .
\end{align*}
Irrespective of whether $\Delta_{kj} = 1$ or $\Delta_{kj} = 0$, under Assumption \ref{assumption p} $\sim$ \ref{assumption sparsity}, the error term is bounded as
$$E(\|Y_j - P_{X_k} Y_j \|_F^2) = O(n^{1+\xi/2}) + O(n^{\tau+\xi/2+\eta/2}) - O(n^{\eta+\xi}) + O(n-p_k) = O(n^{1+\xi/2}).$$
Hence we have:
\begin{align*}
E\dfrac{\left\|P_{X_{k}} Y_{j}\right\|_{F}^{2}}{p_k - 1} = \begin{cases}
 O(n^{1-\xi}) & \text{ if } \Delta_{kj}=1 \\
 O(n^{\eta}) & \text{ if } \Delta_{kj}=0
\end{cases} \quad , \quad
E\dfrac{\left\|Y_{j}-P_{X_{k}} Y_{j}\right\|_{F}^{2}}{n - p_k-1} = O(n^{\xi/2}).
\end{align*}

\subsection{Performance of the indicator function estimator}
The optimization function \eqref{eq deltakj} provides a straightforward rule for estimating $\Delta_{kj}$:
\[
\hat{\Delta}_{kj} =
\begin{cases}
1 & \text{if } \dfrac{\|Y_{j}-P_{X_{k}} Y_{j}\|_{F}^{2}}{n-p_k-1} \le \gamma \dfrac{\|P_{X_{k}} Y_{j}\|_{F}^{2}}{p_k - 1}, \\
0 & \text{if } \dfrac{\|Y_{j}-P_{X_{k}} Y_{j}\|_{F}^{2}}{n-p_k-1} > \gamma \dfrac{\|P_{X_{k}} Y_{j}\|_{F}^{2}}{p_k - 1}.
\end{cases}
\]
To ensure the accuracy of \eqref{eq deltakj} in estimating $\Delta_{kj}$, it suffices to establish the existence of a constant $\gamma$ such that
\[\max_{\Delta_{kj} = 1} l_{kj} \le \gamma \le \min_{\Delta_{k'j'} = 0} l_{k'j'},\]
where
\[l_{kj} = \frac{\|Y_{j}-P_{X_{k}} Y_{j}\|_{F}^{2}}{n -p_k-1} \bigg/ \frac{\|P_{X_{k}} Y_{j}\|_{F}^{2}}{p_k - 1}.\]
Building on the bounds for the projection and error terms discussed in Section \ref{section 2.2}, the asymptotic order of \(l_{kj}\) is derived as follows:
\begin{prop}\label{prop the order of lkj}
Under Assumptions \ref{assumption p} -- \ref{assumption sparsity}, for any $(k,j)$ and $(k',j')$ such that $\Delta_{kj} = 1$ and $\Delta_{k'j'} = 0$, respectively, we have, with probability tending to $1$:
$$l_{kj} = O(n^{3\xi/2-1})  \quad l_{k'j'} = O(n^{\xi/2 - \eta}),$$
where $0<\xi < 1/2$ and $0<\eta + \xi + \tau \le 1$.
\end{prop}

Proposition \ref{prop the order of lkj} can be established using the inequalities and bounds derived in Section \ref{section 2.2}; hence, we omit the detailed proof. Given that $0 < \xi < 1/2$ and $0 < \eta + \xi + \tau \leq 1$, it follows that $\xi/2 - \eta \geq 3\xi/2 - 1$ and $0 \geq 3\xi/2 - 1$. If $\xi/2 - \eta \geq 0$, $l_{kj}$ declines to $0$, whereas $l_{k'j'}$ does not. In this case, there exists a constant $\gamma$ such that $l_{k'j'} > \gamma > l_{kj} \to 0$. On the other hand, if $\xi/2 - \eta \leq 0$, both $l_{kj}$ and $l_{k'j'}$ decline to $0$, but $l_{kj}$ decreases more rapidly. This also implies the existence of a sequence $\gamma$ (dependent on $n$) such that $l_{k'j'} > \gamma > l_{kj}$. Thus, for any specific model, such a constant $\gamma$ can be identified. The following corollary formally states this result:

\begin{coro}\label{coro gamma existence}
Under the same assumptions of Proposition \ref{prop the order of lkj}, for any $(k,j)$ and $(k',j')$ such that $\Delta_{kj} = 1$ and $\Delta_{k'j'} = 0$, there exists a constant $\gamma$ such that, with probability tending to $1$,
$$\max_{\Delta_{kj} = 1} l_{kj} \leq \gamma \leq \min_{\Delta_{k'j'} = 0} l_{k'j'}.$$
\end{coro}

We integrate all blocks and define the following correlation penalized function to select the indicator for the complete block-selection model:
\[Q( \Delta )=
\frac{1}{KJ} \sum_{j=1}^{J} \sum_{k=1}^{K}
\Big(  \frac{1}{n-p_k-1} \left\|Y_{j}-P_{X_{k}} Y_{j}\right\|_{F}^{2} \Delta_{k j}
+\gamma \frac{1}{p_k - 1}\left\|P_{X_{k}} Y_{j}\right\|_{F}^{2}\left(1-\Delta_{k j}\right)
\Big),\]
where $\gamma$ is a tuning parameter. The indicator matrix $\hat \Delta = \{\hat \Delta_{kj}\}_{K \times J}$ is estimated as follows:
\begin{equation}\label{eq Delta}
\hat \Delta =\underset{\Delta \in \mathcal{V}^{K \times J}}{\operatorname{argmin}}~Q(\Delta).
\end{equation}

In the following, we provide the theoretical guarantees for the indicator. First, we introduce the following proposition to ensure the uniqueness of the estimate.

\begin{prop}\label{prop delta}
Given the projection matrix of $X_k$, denoted as $P_{X_k}$ where $k = 1,\dots, K$, the solution of the 0-1 integer optimization problem \eqref{eq Delta} is unique.
\end{prop}

Using Proposition \ref{prop the order of lkj}, we derive the following result on selection consistency.
\begin{thm}\label{thm selection consistency}
Assume Assumptions \ref{assumption p}$ - $\ref{assumption sparsity} holds, with a  properly choose $\gamma$,
the selector is uniformly consistent. Specifically, as $n \rightarrow \infty$, with probability tends to $1$,
\[\hat \Delta  = \Delta.\]
\end{thm}

\rmk{The degrees of freedom for both components play a critical role in ensuring the theoretical guarantees of the correlation-penalized function. Leveraging these degrees of freedom, the correlation function requires only a single tuning parameter, markedly reducing the number of parameters compared to other related methods. A similar perspective has been discussed in \cite{hu2023response}.}

\rmk{When the dimension of covariates within each group is lower than $n$, the results hold under Assumption \ref{assumption p}. However, if the dimension of covariates within some groups exceeds $n$, Assumption \ref{assumption p} no longer applies, and Assumption \ref{assumption p'} is required instead. In such cases, the coefficient estimator within the block should be computed using \eqref{eq esti 2} instead of \eqref{eq esti 1}. The results remain valid provided the active set $\hat S$ is appropriately specified. The properties of the estimator in \eqref{eq esti 2} are discussed in the subsequent section.}

The proposed indicator, along with its associated properties, such as Proposition \ref{prop the order of lkj} and Corollary \ref{coro gamma existence}, demonstrates that for sufficiently large $n$, a constant $\gamma$ can be visually determined such that $l_{kj} < \gamma < l_{k'j'}$. Furthermore, Theorem \ref{thm selection consistency} establishes that the estimator consistently identifies the true active set, ensuring reliable block selection. This approach effectively differentiates relevant blocks from irrelevant ones, considerably reducing computational cost while enabling the identification of relevant blocks by distinguishing the values of the selection function. Furthermore, the non-zero block-indicator detection rule provides a reliable and stable mechanism for identifying non-zero blocks. Considering these aspects, we preliminarily conclude that the correlation-selection function offers an accurate method for computing the selection indicator and facilitating block selection within the model. In the following section, we will delve into the calculation of the tuning parameter $\gamma$.

\subsection{Determination of $\gamma$}\label{section gamma}
The above section provides a selection indicator with a simple rule for block selection. To further discuss the determination of the tuning parameter and provide the calculation for the indicator, we transform \eqref{eq deltakj} into the following function:
\begin{equation}\label{eq deltakj a}
\hat{\Delta}_{kj}=\left\{\begin{array}{ll}
1 & \text { if } \bar{R}_{kj}^2 > 1-\gamma, \\
0 & \text{otherwise},
\end{array}\right.
\end{equation}
where
\[\bar{R}_{kj}^2 \triangleq 1-l_{kj} = 1 - \frac{p_k - 1}{n-p_k-1} \cdot \frac{{\left\|Y_{j}-P_{X_{k}} Y_{j}\right\|_F^{2}}}{{\left\|P_{X_{k}} Y_{j}\right\|_{F}^{2}}}.\]
The determination of $1 - \gamma$ plays an important role in this indicator estimation. For notational simplicity, in this part we denote $c = 1-\gamma$ where $c$ refers to the partition threshold of $\{\bar{R}_{kj}^2:\Delta_{kj} = 1\}$ and $\{\bar{R}_{kj}^2:\Delta_{kj} = 0\}$, which is similar to the partition plane in support vector machines. We aim to choose $c$ such that
\[\begin{array}{ll}
\lim\limits_{n \to \infty} P(\bar{R}_{kj}^2 > c|\Delta_{kj} = 1) \to 1 \\
\lim\limits_{n \to \infty} P(\bar{R}_{kj}^2 \leq c|\Delta_{kj} = 0) \to 1.
\end{array} \]
Set $J_1 = \{(k,j): \Delta_{kj} = 1\}$ and $J_0 = \{ (k,j): \Delta_{kj} = 0 \}$. The related estimates are $\hat J_1$ and $\hat J_0$. Following the idea of the double-thresholding filter from \citep{geng2023large}, we first define two measures, recall rate and error rate (ER), as follows:
\begin{align*}
\text{Recall}(c) &= \frac{|\hat J_1(c) \cap J_1|}{|J_1|},~~  \text{ER}(c) = \frac{|\hat J_1(c)\cap J_0|}{|\hat J_1(c)|} = \frac{|I(c)|}{|\hat J_1(c)|}
\end{align*}
where
\[ \hat J_1(c) = \{(k,j):\hat\Delta_{kj} = 1\} = \{(k,j):\bar{R}_{kj}^2 > c\}, ~\text{and}~ I(c) = \hat J_1(c)\cap J_0.\]
$\hat J_1(c)$ denotes the estimated active set determined by the constant $c$. $J_1$ and $J_0$ denote the true relevant set for blocks and complementary set, respectively. $I(c)$ denotes the Type I error set. To control the recall rate and error rate and obtain the optimal tuning parameter, the following function is presented:
\[\begin{array}{ll}
c = \argmax \text{Recall}(c) \\
\text{s.t.} ~ {\text{ER}}(c)  \leq \alpha,
\end{array} \]
where $\alpha$ is the significant level that usually assumes a value at 0.05 or 0.1. Note that although $J_1$ is unknown, it is deterministic, and the following holds:
$$\argmax \text{Recall}(c) = \argmax \frac{|\hat J_1(c) \cap J_1|}{|J_1|} = \argmax |\hat J_1(c) \cap J_1|.$$
Since $|\hat J_1(c) \cap J_1|$ is monotonically non-increasing in $c$, set $c_0 = \inf \{c:\text{ER}(c) \le \alpha  \}$ and it is the solution of the above optimization function. To properly estimate $\text{ER}(c) = |I(c)|/|\hat J_1(c)|$ and find $c_0$, we first need to estimate $I(c)$, under an additional assumption, we provide theoretical guarantees on the determination of $\gamma$.

\begin{assumption}\label{assumption enhanced symmetry}
When $\Delta_{kj} = 0$, denote $c = (p_k - 1)/(n - p_k - 1)$, and the following terms have asymptotically symmetric distributions with mean 0:
\begin{align*}
&\|P_{X_k}^\perp(Y_j - E_j)\|_F^2 - c \cdot \|P_{X_k}(Y_j - E_j)\|_F^2, \\
&\langle P_{X_k}(Y_j - E_j), P_{X_k}E_j \rangle - c \cdot \langle P_{X_k}^\perp(Y_j - E_j), P_{X_k}^\perp E_j \rangle, \\
&\langle P_{X_k}E_j, P_{X_k}E_j \rangle - c \cdot \langle P_{X_k}^\perp E_j, P_{X_k}^\perp E_j \rangle.
\end{align*}
\end{assumption}

Assumption \ref{assumption enhanced symmetry} models the behavior of quadratic forms involving projections onto $P_{X_k}$ and $P_{X_k}^\perp$. The first term, which compares the squared Frobenius norms, is scaled by a factor $c = \frac{p_k - 1}{n - p_k - 1}$, accounting for the imbalance in degrees of freedom between the two subspaces. The second and third terms assume that the differences between the inner products of the signal and noise projections in the two subspaces are asymptotically symmetric with mean zero. This holds when $E_j$ is uncorrelated with the true signal $Y_j - E_j$, implying a uniform distribution of projections.

Hutchinson’s trick \citep{hutchinson1989stochastic} supports this symmetry assumption, as it shows that quadratic forms involving $E_j$ and $P_{X_k}$ exhibit symmetric behavior under high-dimensional central limit theorems. Specifically, the expectations are:
\[
E(\langle P_{X_k}E_j, P_{X_k}E_j \rangle) = q_j p_k, \quad E(\langle P_{X_k}^\perp E_j, P_{X_k}^\perp E_j \rangle) = q_j (n - p_k),
\]
indicating they are centered around these values. For large $n$, fluctuations diminish, supporting the assumption.
The assumption captures the behavior of projections in high-dimensional settings, with $E_j$ assumed to be uncorrelated with the signal. If $E_j$ exhibits dependencies, further validation would be needed.

\begin{prop}\label{prop gamma}
Under Assumptions \ref{assumption p} to \ref{assumption enhanced symmetry},
and for $0< \gamma < 1$, i.e., $ 0<c = 1-\gamma<1$, the following holds with high probability, :
$$ \mathbb{P}\{(k,j)\in I(c)\}
\le \mathbb{P}\{ \bar R_{kj}^2 < 2c/(c-1) \}
\approx \frac{|\{(k,j):\bar R_{kj}^2 < 2c/(c-1)\}|}{KJ}.$$
\end{prop}

Proposition \ref{prop gamma} is inspired by the double-thresholding filter framework \citep{geng2023large}. It demonstrates that, under the assumptions, the probability of a pair $(k,j)$ belonging to the set $I(c)$ is bounded by the probability of $\bar{R}_{kj}^2$ being below a threshold determined by $c$. While Assumption \ref{assumption enhanced symmetry} may not always strictly hold, the proposition remains valid, as it uses a looser upper bound to estimate the probability of the first error. This ensures that the strategy for selecting $\gamma$ is robust, even when the assumption is relaxed to some degree.

From Proposition \ref{prop gamma}, we have $|I(c)| \le |\{ \bar R_{kj}^2 < 2c/(c-1) \}|$ and
$$ {\text{ER}}(c) = \dfrac{|I(c)|}{|\hat J_1(c)|} \le
\dfrac{|\{ \bar R_{kj}^2 < 2c/(c-1) \}|}{|\{(k,j):\bar R_{kj}^2>c\}|}. $$
Formally, $c$ is obtained by solving the following optimization problem, which has a solution of $c_0 = \inf \{c:\text{ER}(c) \le \alpha  \}$ and can be solved by grid-point search:
\[\begin{array}{cc}
c =  \argmax |\hat J_1(c) \cap J_1| \\
\text{s.t.} \,\,
\dfrac{|\{(k,j):\bar R_{kj}^2 < 2c/(c-1)\}|}{|\{(k,j):\bar R_{kj}^2>c\}|}  \leq \alpha.
\end{array}\]
We simply illustrate the following algorithm for the NBS.
\begin{algorithm}[!h]
\caption{Non-zero Block selector}\label{alg:cap}
\renewcommand{\algorithmicrequire}{\textbf{Input:}}
\renewcommand{\algorithmicensure}{\textbf{Output:}}
\begin{algorithmic}
\Require $X_k,k=1,2,...,K$ ; $Y_j,j=1,2,...,J$
\Ensure Active set $\hat J_1$    

\State 1: Compute $\bar R_{kj}^2$ for each pair $(k,j)$.

\State 2: Set $N = 0$, $i=1$ and $\hat J_1= \emptyset$.

\While{$i \le K*J$}

$c \gets \bar R_{(i)}^2$   \Comment{$R_{(i)}^2$ is the ascending sequence of $\bar R_{kj}^2$}

\If{$\text{ER}(c)\leq \alpha$}
\State \textbf{Return} $\hat J_1 = \hat J_1(c), N = |\hat J_1(c)|$
\EndIf

\State $i \gets i+1$

\EndWhile

\State \textbf{Return} $\hat J_1 = \emptyset, N = 0$

\end{algorithmic}
\end{algorithm}

\section{Joint estimation under various situations}
In this section, we introduce the details of how the block-selection model is jointly estimated. There are two crucial dimensions of the block-selection model. The first dimension is the number of blocks, i.e., $K \times J$, corresponding to the number of response and covariate groups, which we allow to grow with and much higher than $n$. The second dimension is within each block. For $k = 1,\dots, K$ and $j = 1\dots, J$, the dimension of the coefficient submatrix within this block is $p_k \times q_k$, which corresponds to the number of covariates and responses within the block.

For the number of covariates in each group, we consider three situations and propose three estimators for each. First, we consider two types of covariates: 1) $X_k$ is of full rank; and 2) $p_k \gg n$. The optimization of the two situations is different. In the former, we do not consider the sparsity requirement within the block, moreover, the least squared estimation is suitable. In the latter situation, we apply penalized regularization within the block. Ultimately, we consider a joint penalized estimator for the sparse block-selection model.

\subsection{Joint estimation under two kinds of covariates}

Before estimating the block-selection model, we first discuss the estimation of the single-block model and how the indicator reacts in the single-block model. The optimization of estimating the block-selection model is based on the optimization function for estimating the single-block model. Assume $K = 1$ and $J = 1$; thus, the block-selection model is reversed to a regular, multivariate, high-dimensional regression problem:
\begin{equation}\label{eq single model}
Y_1 = X_1 B_{\Delta{11}} + E_1.
\end{equation}
Assume that the dimension of $Y_1$ is $q_1 < n$, while the dimension of $X_1$ is $p_1$. We consider two situations for the dimensionality of $X_1$: $p_1 \leqslant n$ and $p_1 > n$. In this model, the indicator determines the presence of the model. After estimating the single-block model, we discuss the estimate of the block selection, i.e., $K \times J >1$.

\subsubsection{Single-block model under full rank of covariates}

In this part, we consider the covariates to be full rank and do not require sparsity in the model. Set $P_{X_1} = X_1 (X_1^\t X_1)^{-1} X^\t_1$. Following from the above discussion, the selection indicator for the single-block model can be estimated as follows:
\begin{equation}\label{eq delta11}
\hat \Delta_{11} = \argmin_{\Delta_{11} \in \{0,1\}}
\Big(\frac{\left\|Y_1-P_{X_{1}} Y_1\right\|_{F}^{2}}{n-p_1-1} \Delta_{11}
+\gamma \frac{\left\|P_{X_1} Y_1\right\|_{F}^{2}}{p_1-1}\left(1-\Delta_{11}\right)
\Big).
\end{equation}
The estimate of the single-block model can be obtained by transforming \eqref{eq delta11} into the following optimization function:
\begin{equation}\label{eq 11}
(\hat \Delta_{11},\hat B_{\Delta_{11}})= \argmin_{\Delta_{11} \in \{0,1\}} \Big\{\dfrac{1}{n_1}\left\|\Delta_{11} \bullet Y_1 - X_1B_{\Delta_{11}} \right\|_{F}^{2}
+\gamma W_{11} \left(1-\Delta_{11}\right)   \Big\},
\end{equation}
where $n_1 = n-p_1-1$ and $W_{11} = \|  P_{X_1} Y_1 \|_{F}^{2}/(p_1-1)$. The following proposition shows the equivalence of the solutions from \eqref{eq 11} and \eqref{eq deltakj}.
\begin{prop}\label{prop equiv}
A solution to \eqref{eq 11} is equivalent to that to \eqref{eq deltakj} and the regression coefficient estimator is unique, as follows:
\[\hat B_{\Delta 11} = (X_1^\t X_1)^{-1} X^\t_1 Y_1 \bullet \hat \Delta_{11}. \]
\end{prop}

The above result establishes the equivalence between estimating the selection indicator and estimating the coefficient matrix for a single-block model. Furthermore, it introduces the uniqueness of the estimator. The results are natural, and the estimator is the block-selected least squares solution. Subsequently, we provide theoretical guarantees for the estimate, and these can be directly derived from Theorem \ref{thm selection consistency}.

\begin{thm}\label{thm ols}
Assume the Assumptions \ref{assumption p}, \ref{assumption Ej} and \ref{assumption sparsity} holds. With probability tending to 1, the estimate to achieve the oracle solution is as follows:
\[ \hat B_{\Delta 11} = (X_1^\t X_1)^{-1} X^\t_1 Y_1 \bullet \Delta_{11}. \]
\end{thm}

The proof of Theorem \ref{thm ols} is similar to that of Theorem~\ref{thm selection consistency}, with a weaker assumption. Based on the above result, we introduce the asymptotic normality of the coefficient estimator.

\begin{thm}\label{thm normal}
Suppose $E_1$ are sub-Gaussian. For each fixed dimension, $q_0 \leqslant q_1$, set the estimator $\hat B_{ \Delta 11,q_0} = \{\hat \beta^{11}_{ll'}\}_{p_1 \times q_0}\bullet\hat \Delta_{11}$ denoting the first $q_0$ columns of $\hat B_{ \Delta 11}$. Then, the estimator is asymptotically normal, that is,
\[ n^{1/2} (\hat B_{ \Delta 11, q_0} - B_{11,q_0}) \rightarrow N_{p_1 \times q_0}( 0, (X^\t_1 X_1/n)^{-1}\Sigma_{q_0}), \]
where $B_{11,q_0}$ denotes the first $q_0$ columns of $B_{11}$ and $\Sigma_{q_0}$ denotes the covariance of $E_{1,q_0}$.
\end{thm}

\subsubsection{Single-block model under high-dimensional settings}
In this part, we consider the single-block model in which both dimensions of covariates and response are allowed to be larger than the sample size. The covariate, $X_1$, is incapable of full rank and we consider the sparsity assumption. Denote $S_{11}$ to be the set of relevant covariates within the single-block model, and we assume $|S_{11}| < n$. Set
\[P_{X_1} =  X_{\hat S_{11}}(X_{\hat S_{11}}^TX_{\hat S_{11}})^{-1}X_{\hat S_{11}}^T. \]
We use the $l_1$ penalized regularization for the single-block model to obtain the active set, denoted as $\hat S_{11}$ and referred to as lasso for simplicity. As above, the selection indicator can be estimated as follows
\begin{equation}\label{eq delta111}
\hat \Delta_{11} = \argmin_{\Delta_{11} \in \{0,1\}}
\Big(\frac{\left\|Y_1-P_{X_{1}} Y_1\right\|_{F}^{2}}{n-|\hat S_{11}|-1} \Delta_{11}
+\gamma \frac{\left\|P_{X_1} Y_1\right\|_{F}^{2}}{|\hat S_{11}|-1}\left(1-\Delta_{11}\right)
\Big).
\end{equation}
This is equivalent to the following joint estimation of both the correlation penalized function and coefficient estimation,
\begin{align} \label{eq single}
(\hat \Delta_{11},\hat B_{\Delta_{11}})
& = \argmin_{\Delta_{11} \in \{0,1\}, B_{11}: \beta^{11}_{ll'} = 0,  l \notin \hat S_{11}} \Big\{\frac{\left\|\Delta_{11} \bullet Y_1- X_1B_{\Delta_{11}} \right\|_{F}^{2}}{n-|\hat S_{11}|-1}
+\gamma W_{11} \left(1-\Delta_{11}\right)  \Big\},
\end{align}
where $B_{\Delta_{11}} = B_{11}\bullet\Delta_{11} $ and $W_{11} = \|P_{X_1} Y_1\|^2_F/(|\hat S_{11}|-1)$. The following proposition presents the solution of the algorithm and ensures the uniqueness of the estimator.
\begin{prop}\label{prop two}
The solution to \eqref{eq single} is unique and equivalent to that of \eqref{eq delta111}. It can be represented as follows:
\[\hat B_{\Delta 11,\hat S_{11}} = (X_{\hat S_{11}}^TX_{\hat S_{11}})^{-1}X_{\hat S_{11}}^T Y_1 \bullet \hat \Delta_{11}, \]
and $\hat B_{\Delta11, \hat S^c_{11}} = 0$.
\end{prop}

Furthermore, in the single-block model, Assumptions \ref{assumption p'}, \ref{assumption correlation} and \ref{assumption sparsity} can be relaxed. We present the following result to show that the proposed estimate achieves the oracle solution for the single-block model and provide the asymptotic normality of the estimator.
\begin{thm}\label{thm lasso}
Suppose $p_1 = O(\exp^{n^{c_1}})$, $q_1 = o(n^{(1+c_1)/2})$, and the true nonzero set, $|S_{11}| = o(n^{c_2})$, where $0 <c_1 + c_2 < 1$. Assume that the rows of the random error matrix, $E_1$, are independent and identically distributed with mean zero and covariance $\Sigma_1$. Additionally, assume the smallest eigenvalue of $X_1^\t X_1/n$ is larger than $\Lambda_{\min} > 0$ and that the following restricted eigenvalue condition holds:
\[ v^\t (X_1^\t X_1/n) v \geqslant \kappa \|v\|^2_2 ~~\text{for all}~~ v \in G(S_{11},1), \]
where $G(S_{11},1) = \{ v \in \mathbb{R}^{p_1}: \|v_{S^c_{11}}\|_1 \leqslant 3\|v_{S_{11}}\|_1\}$. Furthermore, assume there is a small gap between $0$ and $|B_{\min}| \equiv \min\{|\beta^{11}_{ll'}|, \beta^{11}_{ll'} \neq 0\}$ such that, for a positive constant $M_2$,
\[ |B_{\min}| \geqslant M_2 \cdot \sqrt{n^{c_1 + c_2-1}}. \]
With probability tending to 1, the estimate achieves the oracle solution:
\[ \hat B_{\Delta 11} = (X_{S_{11}}^\t X_{S_{11}})^{-1} X^\t_{S_{11}} Y_{1} \bullet \Delta_{11}. \]
\end{thm}
\rmk{In Theorem \ref{thm lasso}, the assumptions can be seen as special cases of \ref{assumption p'}, \ref{assumption Ej} and \ref{assumption sparsity}. Specifically, we replace Assumption \ref{assumption sparsity} with the restricted eigenvalue condition, which controls sparsity more effectively. Requirements in this theorem are inherited from the lasso framework \citep{negahban2012unified,yyh2020two} as the active set $\hat S_{11}$ is obtained from the lasso.}

\begin{thm}\label{thm lasso normal}
Suppose $E_1$ are sub-Gaussian. Suppose the same assumptions of Theorem~\ref{thm lasso}. For each fixed dimension, $q_0 \leqslant q_1$, set the estimator $\hat B_{ \Delta 11,q_0} = \{\hat \beta^{11}_{ll'}\}_{p_1 \times q_0}\bullet\hat \Delta_{11}$. Then, the estimator is asymptotically normal, that is,
\[ n^{1/2} (\hat B_{ \Delta 11, q_0} - B_{11,q_0}) \rightarrow N_{p_1 \times q_0}( 0, (X^\t_1 X_1/n)^{-1}\Sigma_{q_0}), \]
where $B_{11,q_0}$ denotes the first $q_0$ columns of $B_{11}$ and $\Sigma_{q_0}$ denotes the covariance of $E_{1,q_0}$.
\end{thm}

\subsection{Joint penalized estimation of the sparse block-selection model\label{Discussion}}

As noted in the above two subsections, after estimating the single-block model, the complete block-selection model can be obtained directly by integrating \eqref{eq 11} and \eqref{eq single}:
\[
(\hat \Delta,\hat B_{\Delta})= \argmin_{\Delta \in \mathcal{V}^{K \times J},
\{B_{kj} \in \mathcal{B}_{kj}\}_{K\times J}} \dfrac{1}{KJ} \sum^J_{j = 1} \sum^K_{k = 1} \Big\{\dfrac{1}{n_k}\left\|\Delta_{kj} \bullet Y_j - X_kB_{\Delta_{kj}} \right\|_{F}^{2}
+\gamma W_{kj} \left(1-\Delta_{kj}\right)  \Big\},
\]
where $\mathcal{V} = \{0,1\}$, $\mathcal{B}_{kj} = \{ B_{kj} \in \mathbb{R}^{p_k \times q_j}: \beta_{ll'}^{kj} = 0,l\notin \hat S_{kj}\} $ and $\hat S_{kj}$ equals the complete set when $X_k$ is of full rank.

In this part, we propose an effective algorithm for situations where sparsity is always required. This is a natural situation in which, even if the dimensions of a block are smaller than the sample size and the number of blocks is shrunk by the indicator selector, the total number of coefficients in the complete block model remains relatively large. In this case, we introduce the following joint penalized estimation for the sparse block model.

\begin{algorithm}
\caption{Two-step algorithm for block-selection model\label{algorithm 2}}
\hspace*{0.02in}{\bf Inputs:}$X$ and $Y$.\\\vspace{-25 pt}
\begin{algorithmic}[0]
\State Step 1: for each $k = 1,\dots, K$ and $j = 1,\dots,J$, calculate $\hat \Delta_{kj}$ based on \eqref{eq 11} or \eqref{eq single}.

\State Step 2: Calculate the coefficient matrix $B_\Delta$ as follows:
\[\hat B_{\Delta}=\underset{B_{\Delta_{kj}} = 0, (k,j) \in \hat J_0}{\operatorname{argmin}}\left\{\|Y-X B_{\Delta}\|_{F}^{2}+\lambda \|B_{\Delta}\|_1\right\}. \]
\end{algorithmic}
\hspace*{0.02in}{\bf Output:} $\hat B_{\Delta}$.
\end{algorithm}

This method enjoys the computational advantage of traditional penalized optimization techniques. After determining the selection indicator, solving the block model is equivalent to solving a lasso-type optimization, which significantly reduces the computational cost. Based on the above discussions, we have derived the following two properties for the coefficient estimator.

\begin{thm}\label{thm 2}
Assume the following restricted eigenvalue condition holds:
\[ v^\t(X^\t X/n) v \geqslant \kappa \|v\|^2_2,\]
for all $v \in G(S)$ where $G(S) = \{ v \in \mR^P : \|v_{S^c}\|_1 \leqslant 3\|v_{S}\|_1\}$. Assume $\lambda =  M_3 \sqrt{n\log P}$ where $M_3$ is a positive constant. Then, with probability tending to $1$, the following error bound for the estimate holds,
\[\|\hat B - B\|_2 \leqslant M_3 \sqrt{\dfrac{|S|\log P}{n}},\]
where $S = \cup S_{kj} $.
\end{thm}

Similarly to the notation of Theorem~\ref{thm lasso}, we assume $B_{\min} = \min\{ |\beta^{kj}_{ll'}| : \beta^{kj}_{ll'} \neq 0, k= 1,\dots,K, j = 1,\dots,J \}$ and the following result hold,
\begin{thm}\label{thm 3}
Under the same assumptions of Theorem~\ref{thm 2}. Assume $B_{\min} \geqslant M_3 \sqrt{|S| \log P /n}$. Then, we have
\[ P(\sign(\hat B) = \sign(B)) \rightarrow 1.\]
\end{thm}

The above two results provide overall theoretical guarantees for the block-selection model. We do not restrict the number of blocks or the dimensions within each block. Specifically, we allow the dimensions of each block to grow exponentially. In this case, an effective and robust shrinkage technique is crucial to our method. We name the joint estimation obtained by Algorithm \ref{algorithm 2} NBSlasso and demonstrate its performance in simulations and a real example below.

\section{Simulation}
In this section, we compare the performance of NBSlasso, SCCS, lasso, and elastic net through numerical simulation experiments, where NBSlasso is obtained using Algorithm~\ref{algorithm 2}. We conduct numerical simulation experiments with two different data dimensions: $(n=150,p=200,q=200)$ and $(n=150,p=400,q=200)$, as well as the following two group settings:

\textbf{Group setting 1.} Each group has the same size of $20$, i.e., under the dimension $(n=150,p=200,q=200)$, the group structures is $(1-20,21-40,41-60,...,181-200)$.

\textbf{Group setting 2.} The group structures under dimension $(n=150,p=200,q=200)$ are represented by $(1-20,21-50,51-70,70-100,...,171-200)$, in which all the group sizes are unequal.

We generate the covariates, $X$, from a multivariate normal distribution, $N_p(0,\Sigma)$, where $\Sigma = \{0.5^{|k - k'|}\}_{p \times p}$. Recall that we do not demand independence among the various covariate groups; rather, we permit a modest level of correlation to better approximate the true structure of real datasets. Initially, we designate $K_j$ as the sparsity indicator for $Y_j$. For example, $K_1 = 2$ implies that we randomly select two nonzero blocks from $B_{11},\dots,B_{K1}$ for $Y_1$. The sparsity level $K_j$ is determined in two scenarios: 1) The fixed case. $K_j$ is constant for $Y_j$ where $j = 1,\dots,J$; 2) The random case. $K_j$ is randomly chosen and may vary among different response groups. Within each nonzero block, the sparsity level of the sub-coefficient matrix is selected from $10\%$ to $90\%$, which means that 10\% to 90\% of the entries in the block are zero. The values of nonzero entries are from a uniform distribution, $[-5,-1]\cup [1,5]$. Figure \ref{simulation dataset gen} illustrates the process of generating a simulation data set and all simulations are repeated 100 times.

\begin{figure}[H]
\centering
\includegraphics[scale=0.12]{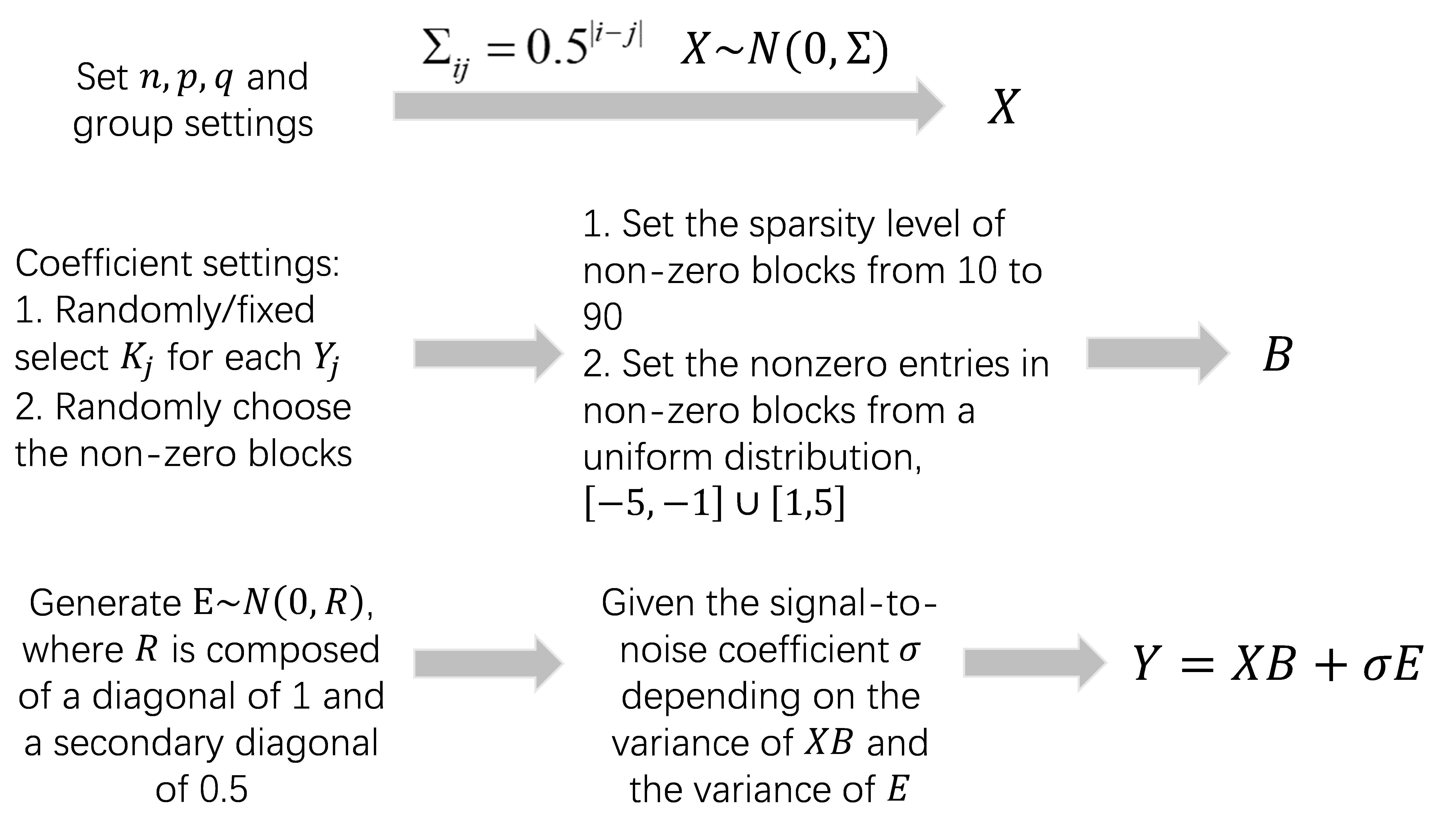}
\caption{Generating simulation dataset.\label{simulation dataset gen}}
\end{figure}

We compare the performance of the four methods using the following nine metrics: the mean square error of the test set (TestMSE); the precision of selected non-zero blocks (Precision); the recall rate of non-zero blocks (Recall); the $L_1$ norm of $B-\hat{B}$; the $L_2$ norm of $B-\hat{B}$; the positive discovery rate of non-zero entries (PDR); the false discovery rate of non-zero entries (FDR); the number of non-zero entries estimated by the methods (NNE); and the computational time (Time). Note that ER, as mentioned in Section \ref{section gamma}, is equal to $1-\text{Precision}$. The simulations are performed on a standard laptop computer equipped with a 2.60 GHz Intel Core i5-11260H processor.

\vspace{0.8em}
\noindent{\textbf{Performance comparison with random sparsity indicators}}

In this part, we discuss the performance of the four methods in two scenarios with random sparsity indicators. For each response block, the sparsity indicator $K_j$ is randomly chosen from two sets: $K_j \in \{2,3,4\}$ and $K_j \in \{1,2,...,9\}$. Table \ref{D1G1Random} demonstrates the performance of the four methods under the following conditions: dimension $n=150, p=200, q=200$, Group setting 1, with a random sparsity indicator of $K_j \in \{2,3,4\}$. We consider three sparsity levels here: 30\%, 60\%, and 90\%. As one can see, NBSlasso outperforms the other methods in terms of TestMSE and consistently maintains a low value. It precisely selects all non-zero blocks, achieving satisfactory precision and recall rates of 1. The accuracy of NBSlasso is also reflected in the $L_1$, $L_2$, and PDR, which show the lowest values for $L_1$ and $L_2$ and a high PDR. Additionally, NBSlasso outperforms the other three methods in terms of computational efficiency. NBSlasso yields a higher FDR than SCCS, for example, when the sparsity equals to 90\%, NBSlasso achieves a 0.67 FDR compared to SCCS's 0.29. However, NBSlasso has a precision and recall rate of 1. These two observations suggest that NBS can accurately select all non-zero blocks, but the non-zero entries selected by Lasso contains many spurious variables. Though we do not show the NNE metrics here due to the size limit of the table, the following discussion will involve them. Due to space constraints, additional comparisons under different dimensions and group settings are shown under ``Supplementary.''

\begin{table}[h!]
\centering
\caption{Performance comparison under dimensions $n=150,p=200,q=200$, Group setting 1, random sparsity indicators of $K_j \in \{2,3,4\}$, and sparsity-levels within each block are $30\%, 60\%$, and $90\%$ respectively. \label{D1G1Random}}
\resizebox{\textwidth}{!}{%
\begin{tabular}{cccccccccc}
\hline
Sparsity & Method\_name & TestMSE    & Precision  & Recall     & L1            & L2           & PDR        & FDR        & Time(s)          \\ \hline
30       & NBSlasso     & 0.09(0.00) & 1.00(0.00) & 1.00(0.00) & 345.22(31.79) & 17.68(1.78)  & 0.97(0.00) & 0.22(0.00) & 13.74(1.64)   \\ \cline{2-10}
& SCCS         & 0.50(0.03) & 0.99(0.02) & 1.00(0.00) & 964.44(41.82) & 139.03(6.29) & 0.27(0.02) & 0.06(0.00) & 405.82(6.82)  \\ \cline{2-10}
& ElasticNet       & 0.26(0.03) & 0.31(0.02) & 1.00(0.00) & 828.03(55.17) & 73.69(6.85)  & 0.76(0.03) & 0.55(0.02) & 23.63(0.56)   \\ \cline{2-10}
& Lasso        & 0.26(0.03) & 0.31(0.02) & 1.00(0.00) & 813.51(59.82) & 73.77(7.65)  & 0.73(0.03) & 0.54(0.02) & 23.75(0.72)   \\ \hline
60       & NBSlasso     & 0.08(0.00) & 1.00(0.00) & 1.00(0.00) & 237.81(19.75) & 11.84(0.94)  & 0.99(0.00) & 0.42(0.01) & 11.59(1.38)   \\ \cline{2-10}
& SCCS         & 0.33(0.04) & 0.94(0.05) & 1.00(0.00) & 542.34(45.96) & 84.58(8.31)  & 0.47(0.03) & 0.11(0.00) & 430.47(18.89) \\ \cline{2-10}
& ElasticNet       & 0.11(0.01) & 0.30(0.02) & 1.00(0.00) & 467.10(29.73) & 29.21(3.06)  & 0.94(0.01) & 0.70(0.02) & 19.54(0.86)   \\ \cline{2-10}
& Lasso        & 0.10(0.01) & 0.30(0.02) & 1.00(0.00) & 427.66(29.46) & 25.57(3.04)  & 0.94(0.01) & 0.68(0.02) & 19.16(0.86)   \\ \hline
90       & NBSlasso     & 0.09(0.01) & 1.00(0.00) & 1.00(0.00) & 93.85(10.07)  & 4.98(0.44)   & 1.00(0.00) & 0.67(0.01) & 10.59(1.28)   \\ \cline{2-10}
& SCCS         & 0.10(0.01) & 0.77(0.10) & 1.00(0.00) & 98.17(16.02)  & 11.62(2.64)  & 0.89(0.02) & 0.29(0.05) & 282.64(29.85) \\ \cline{2-10}
& ElasticNet       & 0.06(0.00) & 0.31(0.03) & 1.00(0.00) & 228.65(8.75)  & 8.94(0.61)   & 1.00(0.00) & 0.90(0.01) & 13.77(0.48)   \\ \cline{2-10}
& Lasso        & 0.06(0.00) & 0.31(0.03) & 1.00(0.00) & 202.66(8.10)  & 7.43(0.56)   & 1.00(0.00) & 0.89(0.01) & 12.95(0.45)   \\ \cline{2-10}
\end{tabular}%
}
\end{table}

We then test the performance of all the methods under the more general random sparsity indicator of $K_j \in \{1,2,...,9\}$. Figure \ref{general random case figure2} illustrates the differences between the methods by focusing on four metrics: TestMSE, NNE (the number of estimated non-zero entries), $L_2$, and $L_1$ with sparsity levels ranging from 10\% to 90\% under the dimensions $n=150,p=200,q=200$, Group setting 1, with random sparsity indicators of $K_j \in \{1,2,\dots,9\}$. We can see that NBSlasso generally outperforms the other methods in terms of TestMSE. Regarding the other three metrics, we see that NBSlasso consistently maintains its advantage in terms of the $L_1$ norm and $L_2$ norm, while accurately estimating the number of non-zero coefficients, indicating its accuracy.

\begin{figure}[H]
\centering
\includegraphics[width = 16cm]{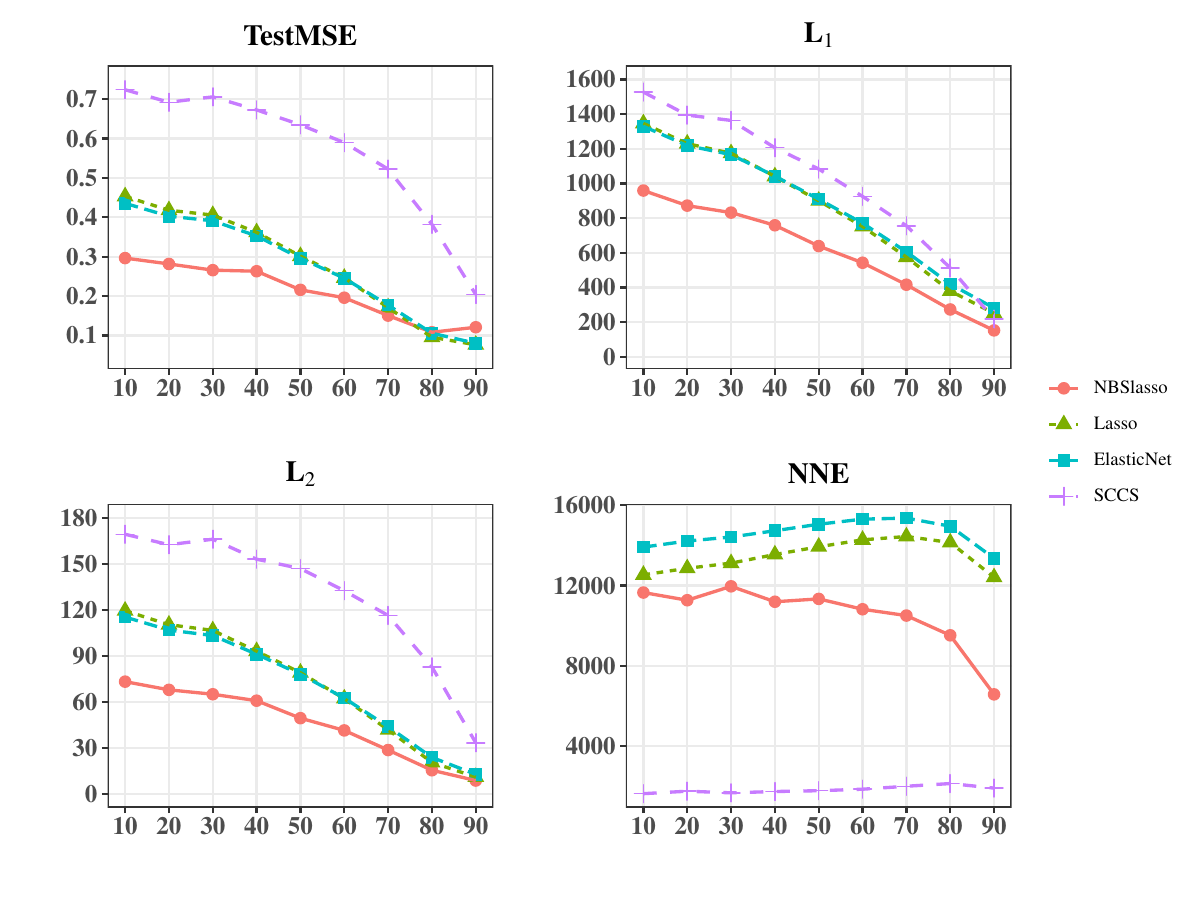}
\caption{Performance comparison under dimensions $n=150,p=200,q=200$, Group setting 1, and random sparsity indicators of $K_j \in \{1,2,\dots,9\}$. The x-axis represents the sparsity level (\%) within non-zero blocks.\label{general random case figure2}}
\end{figure}

\vspace{0.4em}
\noindent{\textbf{Performance comparison with fixed sparsity indicators}}

In this part, we demonstrate the performance of all the methods with fixed sparsity indicators of $K_j = 2,4,6,8$, respectively. In Figure \ref{general random case figure3}, we display the TestMSE, NNE, and $L_2$ metrics of all the methods with sparsity levels ranging from 10\% to 90\% under dimensions $n=150,p=200,q=200$ and Group setting 1. In terms of $K_j = 2$, the NNE estimated by NBSlasso is closest to the real model, and NBSlasso has the smallest $L_2$. Though in the cases with high sparsity levels, the TestMSE of NBSlasso is larger than that of lasso or elastic net, it is most likely due to the overfitting of lasso and elastic net, which can be inferred from the much larger NNE than the real model. The results for $K_j = 4$ are similar to those for $K_j = 2$. In the denser cases, such as when $K_j = 6$ and $K_j = 8$, the NNE estimated by NBSlasso is not close to the real model, like other methods. Though the precision and recall rate of non-zero blocks are not presented here, we claim that the bias of NNE between NBSlasso and the real model is not due to block selection but rather to coefficient estimation by the penalized method. Nonetheless, NBSlasso maintains the best $L_2$ metrics, indicating that the model estimated by NBSlasso is more accurate than others. To conclude, NBSlasso outperforms the other methods in terms of TestMSE in most situations. Additionally, the NBS has the smallest $L_2$ while estimating fewer non-zero entries compared to the other methods.

\begin{figure}[htbp]
\subfigure{
\includegraphics[width = 16cm]{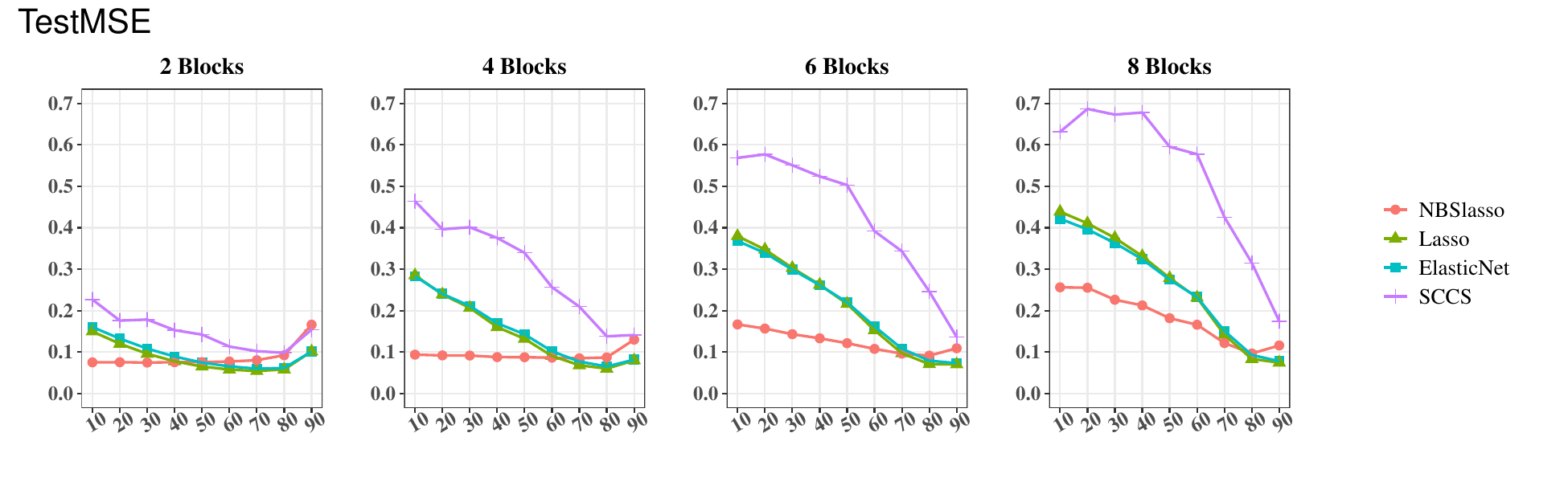}}
\subfigure{
\includegraphics[width = 16cm]{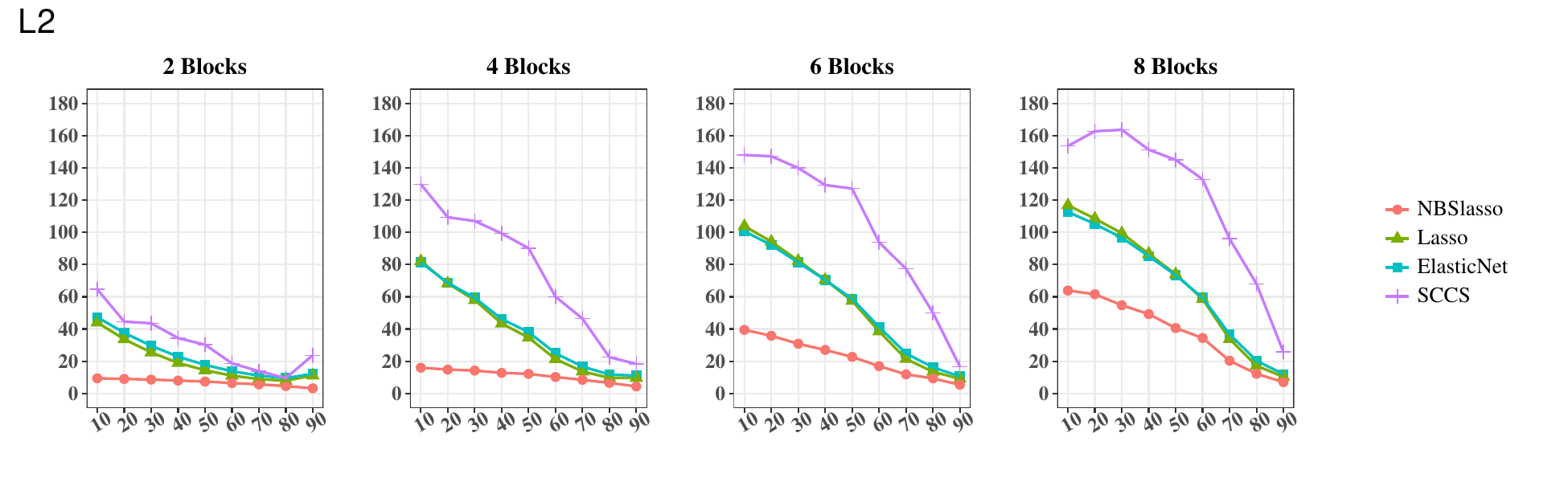}}
\subfigure{
\includegraphics[width = 16cm]{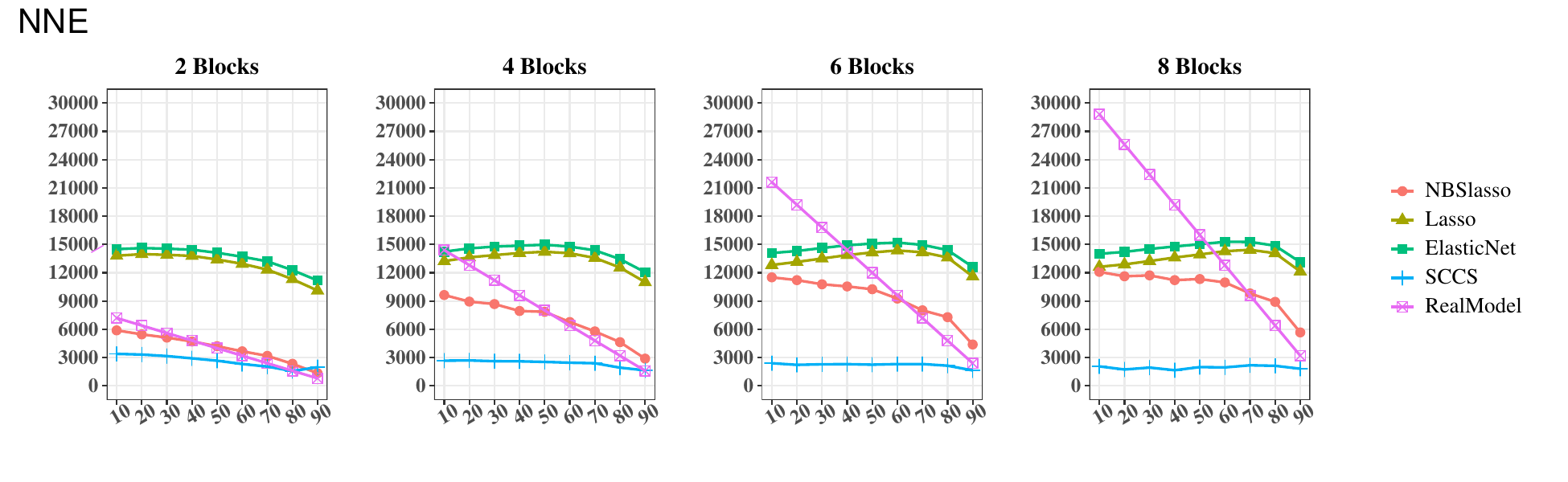}}
\caption{
Performance comparison under dimensions $n=150,p=200,q=200$ and Group setting 1. The title of each figure in each row, indicates the fixed sparsity indicators, $K_j = 2, 4, 6, 8$, respectively. The x-axis represents the sparsity level (\%) within non-zero blocks. \label{general random case figure3}}
\end{figure}

In summary, NBSlasso consistently outperforms the other methods in various scenarios. It demonstrates superior performance in terms of computation time, estimation accuracy, and TestMSE. These results hold regardless of whether the data are sparse or dense and whether the number of non-zero coefficient blocks is randomly set or fixed. The findings indicate that NBSlasso is a robust and effective method for selecting informative blocks and accurately estimating the associated coefficient matrix.

\section{Real example}
In this section, we present the performance of our proposed method on a practical problem of genomic study, which is the interaction between the isoprenoid genes in Arabidopsis thaliano and has been studied in \citet{wille2004sparse}. Isoprenoids serve numerous biochemical functions in plants and are synthesized through the condensation of the five-carbon intermediates isopentenyl diphosphate (IPP) and dimethylallyl diphosphate (DMAPP). The formation of IPP and DMAPP in higher plants has two distinct pathways: one in the cytosol, the mevalonatepatty pathway, and the other in the chloroplast, nonmevalonatepatty pathway. The interaction between these two pathways has been reported by many researchers \citep{laule2003crosstalk, rodrieguez2004distinct}. We examine a data set that comprises 118 GeneChip (Affymetrix) microarrays and contains the expression of 39 genes in the isoprenoid pathways in Arabidopsis thaliano and 795 additional genes from 56 downstream metabolic pathways \citep{wille2004sparse}. The 39 genes in the isoprenoid pathways can be divided into two groups: the mevalonatepatty pathway group containing 21 genes and the nonmevalonatepatty group containing the remaining 18 genes. \citet{wille2004sparse} utilized this data set to construct the genetic regulatory network between the isoprenoid pathways where the downstream genes were attached to the network as conditional variables.

The construction of the genetic regulatory network can be solved by formulating a conditional Gaussian graphical model with group structures $Y = XB + E$. There exists multicollinearity in the 795 additional genes from 56 downstream metabolic pathways and we remove some downstream metabolic pathways, obtaining a model as $(Y_1, Y_2) = (X_1,\cdots, X_{50})B + E$, where $(n, p, q) = (118, 739, 39), J=2$, and $K=50$.
$(Y_1, Y_2)$ represents the isoprenoid genes in the mevalonatepatty and nonmevalonatepatty pathways, and $(X_1,\cdots, X_{50})$ denotes the genes from 50 downstream metabolic pathways. We focus on the selection of the downstream pathways as predictors for the isoprenoid pathways and compare NBSlasso, lasso, elastic net, and SCCS using four metrics: TestMSE (the mean square error in the testing set), NNB (the number of non-zero blocks estimated), NNE (the number of non-zero entries estimated), and computational time.

We set the size of the training set as $n_{train} = (100,90,80,70,60)$. For each $n_{train}$, we randomly split the data set into a training set and a testing set. We estimate $B$ using four methods and compute the TestMSE on the testing set. Each test is repeated 100 times on different splits. The results are shown in table \ref{realeg}. Firstly, we compare the NBSlasso with the SCCS. We observe that the average NBSlasso achieves a lower TestMSE and a greater NNE. This implies that SCCS might overlook some true non-zero entries. Although NBSlasso selects more accurate non-zero entries, it sppears to select significantly more non-zero blocks than SCCS, which is likely attributable to presence of multicollinearity in the predictors. NBSlasso still maintains its computational advantage over SCCS. Secondly, we focus on the differences between NBSlasso and the regularized regression approaches. Both Lasso and elastic net selects almost all the blocks while NBSlasso selects far fewer. NBSlasso obtains a smaller TestMSE with fewer NNE than Lasso and elastic net, indicating that the regularized methods may suffer from the over-fitting problem. The computational time of NBSlasso is comparable to that of the Lasso and elastic net. To conclude, the results show that the NBSlasso detects a more accurate gene set with less computational cost.

\begin{table}[!ht]
\centering
\caption{The mean value of four metrics of the real example (Values in parentheses are standard deviations).\label{realeg}}
\resizebox{\textwidth}{!}{%
\begin{tabular}{cccccccccccc}
\multicolumn{1}{l}{} & \multicolumn{1}{l}{} & \textbf{TestMSE} & \multicolumn{1}{l}{} & \multicolumn{1}{l}{} & \multicolumn{1}{l}{} & \multicolumn{1}{l}{} & \multicolumn{1}{l}{} & \multicolumn{1}{l}{} & \textbf{NNB}  & \multicolumn{1}{l}{} & \multicolumn{1}{l}{} \\ \hline
$n_{train}$                    & NBSlasso             & Lasso            & ElasticNet                  & SCCS                 &                      &                      & $n_{train}$                    & NBSlasso             & Lasso         & ElasticNet                  & SCCS                 \\ \hline
100         & 0.60(0.54)    & 0.62(0.57)                        & 0.62(0.57)     & 0.91(0.83)   &  &  & 100         & 53.94(3.47) & 92.07(2.39)                       & 94.65(1.75) & 4.65(1.79)    \\ \hline
90          & 0.63(0.41)    & 0.63(0.42)                        & 0.63(0.41)     & 0.8(0.51)    &  &  & 90          & 53.62(4.41) & 91.17(2.27)                       & 94.06(2.09) & 3.7(1.69)     \\ \hline
80          & 0.76(0.67)    & 0.79(0.68)                        & 0.78(0.68)     & 1.15(0.98)   &  &  & 80          & 53.09(7.63) & 89.8(2.32)                        & 92.85(2.35) & 3.5(1.32)     \\ \hline
70          & 0.85(0.68)    & 0.89(0.66)                        & 0.87(0.66)     & 1.16(0.68)   &  &  & 70          & 46.11(7.48) & 88.84(2.9)                        & 92.51(2.73) & 2.75(1.25)    \\ \hline
60          & 1.0(0.71)     & 1.0(0.77)                         & 0.99(0.77)     & 1.06(0.39)   &  &  & 60          & 43.71(11.4) & 86.92(3.37)                       & 91.15(2.9)  & 2.65(1.09)    \\ \hline
&                      & \textbf{NNE}     &                      &                      &                      &                      &                      &                      & \textbf{Time(s)} &                      &                      \\ \hline
$n_{train}$                    & NBSlasso             & Lasso            & ElasticNet                  & SCCS                 &                      &                      & $n_{train}$                    & NBSlasso             & Lasso         & ElasticNet                  & SCCS                 \\ \hline
100         & 642.4(73.47)  & 889.0(85.59)                      & 1154.17(85.44) & 231.4(36.85) &  &  & 100         & 5.55(1.19)  & 3.7(0.19)                         & 3.6(0.15)   & 142.03(21.73) \\ \hline
90          & 578.35(78.24) & 811.65(73.31)                     & 1079.11(84.65) & 202.1(32.09) &  &  & 90          & 5.23(1.03)  & 3.44(0.15)                        & 3.38(0.15)  & 132.99(21.19) \\ \hline
80          & 531.24(76.7)  & 733.04(73.54)                     & 994.5(86.21)   & 177.95(36.8) &  &  & 80          & 4.62(1.05)  & 3.2(0.18)                         & 3.17(0.16)  & 108.49(21.41) \\ \hline
70          & 502.6(80.91)  & 665.73(87.15)                     & 921.1(99.38)   & 154.9(29.52) &  &  & 70          & 4.74(0.42)  & 3.42(0.15)                        & 3.52(0.19)  & 95.78(17.37)  \\ \hline
60          & 425.69(72.61) & 569.07(80.56)                     & 824.09(94.91)  & 137.45(29.7) &  &  & 60          & 3.81(0.42)  & 2.82(0.13)                        & 3.22(0.11)  & 93.63(24.37) \\ \hline
\end{tabular}%
}
\end{table}

\section{Conclusion}
In this paper, we focused on a block-selection model and introduced a technique called the NBS to identify the relevance between response and covariate groups. We constructed a high-dimensional model with block structures for both responses and covariates, and the new strategy played a crucial role in detecting these blocks. We demonstrated the uniform consistency of the NBS and proposed three estimators to enhance modeling efficiency under different scenarios. Additionally, we derived their asymptotic properties. The results from both simulations and empirical analysis illustrate that the proposed method effectively addresses the challenges of data complexity, achieving satisfactory estimation and prediction accuracy.

We proposed a novel approach, called the NBS, to investigate the blocking (grouping) effect by introducing an efficient correlation measure. The NBS can be applied in combination with various other regularization techniques for diverse modeling, thus offering extensive application prospects. In further research, several topics are worth exploring, such as utilizing the NBS to detect change points in linear models \citep{li2021multithreshold}, reducing estimation bias caused by hidden variables \citep{bing2022adaptive}, etc.

\section*{Acknowledgments}
This work was supported by the National Natural Science Foundation of China (Grant No. 12371281); the Emerging Interdisciplinary Project, the Fundamental Research Funds, and the Disciplinary Funds in Central University of Finance and Economics.

\appendix
\section*{Appendix A: Additional Simulation Results}

Table \ref{D1G2Random}, \ref{D2G1Random}, and \ref{D2G2Random} compare the performance of the four methods under two dimensions, two group settings, random sparsity indicators of $K_j \in \{2,3,4\}$, and sparsity-levels within each block are $30\%, 60\%$, and $90\%$ respectively.

\begin{table}[h!]
\centering
\caption{Performance comparison under dimensions $n = 150, p = 200, q = 200$, Group
setting 2, random sparsity indicators of $K_j \in \{2,3,4\}$, and sparsity-levels within each
block are 30\%, 60\%, and 90\% respectively.\label{D1G2Random}}
\resizebox{\textwidth}{!}{%
\begin{tabular}{cccccccccc}
\hline
Sparsity & Method\_name & TestMSE    & Precision  & Recall     & L1             & L2            & PDR        & FDR        & Time(s)           \\ \hline
30       & NBS     & 0.12(0.05) & 0.97(0.10) & 1.00(0.00) & 488.38(152.10) & 30.17(15.13)  & 0.93(0.06) & 0.23(0.04) & 16.36(2.97)    \\ \cline{2-10}
& SCCS         & 0.95(0.09) & 1.00(0.00) & 0.22(0.33) & 1311.35(53.15) & 193.63(4.67)  & 0.03(0.03) & 0.04(0.02) & 102.39(114.56) \\ \cline{2-10}
& ElasticNet       & 0.31(0.04) & 0.37(0.04) & 1.00(0.00) & 951.62(98.32)  & 87.17(10.45)  & 0.70(0.05) & 0.51(0.04) & 24.37(0.90)    \\ \cline{2-10}
& Lasso        & 0.32(0.05) & 0.37(0.04) & 1.00(0.00) & 946.62(106.19) & 89.00(11.74)  & 0.67(0.05) & 0.49(0.04) & 24.79(1.05)    \\ \hline
60       & NBS     & 0.09(0.01) & 0.97(0.07) & 1.00(0.00) & 310.74(59.08)  & 16.16(3.89)   & 0.98(0.01) & 0.43(0.02) & 14.37(2.69)    \\ \cline{2-10}
& SCCS         & 0.56(0.21) & 0.99(0.02) & 0.81(0.35) & 772.75(201.54) & 131.96(34.49) & 0.27(0.14) & 0.11(0.04) & 416.76(188.74) \\ \cline{2-10}
& ElasticNet       & 0.15(0.03) & 0.37(0.04) & 1.00(0.00) & 562.59(67.27)  & 39.75(7.53)   & 0.90(0.03) & 0.65(0.03) & 20.72(0.93)    \\ \cline{2-10}
& Lasso        & 0.14(0.03) & 0.37(0.04) & 1.00(0.00) & 525.99(70.41)  & 36.54(8.04)   & 0.90(0.03) & 0.64(0.03) & 20.66(1.13)    \\ \hline
90       & NBS     & 0.08(0.01) & 0.95(0.12) & 1.00(0.00) & 121.68(17.66)  & 5.79(0.53)    & 1.00(0.00) & 0.71(0.04) & 12.14(1.27)    \\ \cline{2-10}
& SCCS         & 0.15(0.02) & 0.95(0.03) & 0.96(0.03) & 128.92(22.38)  & 21.24(4.71)   & 0.77(0.04) & 0.14(0.01) & 356.60(11.75)  \\ \cline{2-10}
& ElasticNet       & 0.06(0.00) & 0.38(0.03) & 1.00(0.00) & 237.23(6.91)   & 9.13(0.46)    & 1.00(0.00) & 0.88(0.01) & 13.82(0.47)    \\ \cline{2-10}
& Lasso        & 0.05(0.00) & 0.38(0.03) & 1.00(0.00) & 209.71(7.29)   & 7.54(0.44)    & 1.00(0.00) & 0.87(0.01) & 13.38(0.56)    \\ \hline
\end{tabular}%
}
\end{table}

\begin{table}[h!]
\centering
\caption{Performance comparison under dimensions $n = 150, p = 400, q = 200$, Group
setting 1, random sparsity indicators of $K_j \in \{2,3,4\}$, and sparsity-levels within each
block are 30\%, 60\%, and 90\% respectively. \label{D2G1Random}}
\resizebox{\textwidth}{!}{%
\begin{tabular}{cccccccccc}
\hline
sparsity & method\_name & testMSE    & precision  & recall     & L1             & L2           & PDR        & FDR        & Time(s)          \\ \hline
30       & NBS     & 0.11(0.01) & 1.00(0.00) & 1.00(0.00) & 360.19(37.41)  & 19.40(2.41)  & 0.96(0.01) & 0.20(0.00) & 13.19(1.63)   \\ \cline{2-10}
& SCCS         & 0.54(0.03) & 0.99(0.01) & 1.00(0.00) & 1004.12(70.86) & 145.32(9.62) & 0.25(0.03) & 0.05(0.00) & 375.43(23.17) \\ \cline{2-10}
& ElasticNet       & 0.49(0.04) & 0.16(0.01) & 1.00(0.00) & 1042.48(63.96) & 123.01(7.65) & 0.51(0.03) & 0.60(0.02) & 33.70(0.62)   \\ \cline{2-10}
& Lasso        & 0.49(0.04) & 0.16(0.01) & 1.00(0.00) & 1027.38(67.62) & 123.89(8.38) & 0.47(0.04) & 0.57(0.02) & 33.31(0.73)   \\ \hline
60       & NBS     & 0.09(0.00) & 1.00(0.00) & 1.00(0.00) & 235.70(24.42)  & 12.32(1.35)  & 0.98(0.00) & 0.39(0.01) & 11.43(1.50)   \\ \cline{2-10}
& SCCS         & 0.33(0.05) & 0.98(0.02) & 1.00(0.00) & 536.66(73.78)  & 84.35(13.63) & 0.46(0.05) & 0.10(0.01) & 420.40(24.21) \\ \cline{2-10}
& ElasticNet       & 0.28(0.04) & 0.15(0.01) & 1.00(0.00) & 651.78(55.25)  & 68.13(8.71)  & 0.78(0.03) & 0.72(0.02) & 29.70(0.93)   \\ \cline{2-10}
& Lasso        & 0.26(0.04) & 0.15(0.01) & 1.00(0.00) & 601.59(59.07)  & 61.89(9.50)  & 0.77(0.04) & 0.69(0.02) & 28.67(1.03)   \\ \hline
90       & NBS     & 0.09(0.01) & 1.00(0.00) & 1.00(0.00) & 81.37(9.45)    & 4.90(0.48)   & 1.00(0.00) & 0.58(0.01) & 10.60(1.55)   \\ \cline{2-10}
& SCCS         & 0.11(0.01) & 0.90(0.05) & 1.00(0.00) & 78.16(15.87)   & 9.92(2.69)   & 0.89(0.02) & 0.18(0.01) & 241.01(24.90) \\ \cline{2-10}
& ElasticNet       & 0.08(0.00) & 0.15(0.02) & 1.00(0.00) & 227.34(10.68)  & 10.42(0.78)  & 1.00(0.00) & 0.90(0.01) & 20.06(0.48)   \\ \cline{2-10}
& Lasso        & 0.07(0.00) & 0.15(0.02) & 1.00(0.00) & 186.20(8.53)   & 7.30(0.48)   & 1.00(0.00) & 0.89(0.01) & 18.92(0.57)   \\ \hline
\end{tabular}%
}
\end{table}

\newpage

\begin{table}[h!]
\centering
\caption{Performance comparison under dimensions $n = 150, p = 400, q = 200$, Group
setting 2, random sparsity indicators of $K_j \in \{2,3,4\}$, and sparsity-levels within each
block are 30\%, 60\%, and 90\% respectively.  \label{D2G2Random}}
\resizebox{\textwidth}{!}{%
\begin{tabular}{cccccccccc}
\hline
sparsity & method\_name & testMSE    & precision  & recall     & L1             & L2            & PDR        & FDR        & Time(s)           \\ \hline
30       & NBS     & 0.36(0.14) & 0.50(0.30) & 1.00(0.00) & 900.86(290.27) & 90.75(39.40)  & 0.63(0.19) & 0.43(0.14) & 24.93(5.74)    \\ \cline{2-10}
& SCCS         & 0.94(0.12) & 1.00(0.00) & 0.23(0.34) & 1277.16(77.18) & 189.25(7.00)  & 0.03(0.05) & 0.04(0.01) & 90.56(105.62)  \\ \cline{2-10}
& ElasticNet       & 0.53(0.04) & 0.19(0.02) & 1.00(0.00) & 1145.51(81.08) & 131.06(6.61)  & 0.45(0.04) & 0.57(0.02) & 35.94(1.15)    \\ \cline{2-10}
& Lasso        & 0.54(0.04) & 0.19(0.02) & 1.00(0.00) & 1136.09(84.93) & 133.16(7.32)  & 0.41(0.04) & 0.55(0.02) & 35.21(1.11)    \\ \hline
60       & NBS     & 0.22(0.07) & 0.49(0.29) & 1.00(0.01) & 555.79(168.86) & 49.18(20.87)  & 0.83(0.09) & 0.56(0.10) & 21.83(4.65)    \\ \cline{2-10}
& SCCS         & 0.71(0.26) & 1.00(0.00) & 0.51(0.47) & 819.24(184.98) & 150.48(37.91) & 0.18(0.17) & 0.07(0.02) & 232.69(193.38) \\ \cline{2-10}
& ElasticNet       & 0.35(0.05) & 0.19(0.02) & 1.00(0.00) & 772.40(77.91)  & 84.33(9.99)   & 0.69(0.05) & 0.68(0.02) & 32.41(1.52)    \\ \cline{2-10}
& Lasso        & 0.34(0.05) & 0.19(0.02) & 1.00(0.00) & 731.63(84.22)  & 80.25(11.48)  & 0.68(0.06) & 0.65(0.03) & 31.26(1.24)    \\ \hline
90       & NBS     & 0.08(0.01) & 0.39(0.21) & 1.00(0.00) & 161.61(23.37)  & 7.18(0.82)    & 1.00(0.00) & 0.81(0.07) & 15.26(1.51)    \\ \cline{2-10}
& SCCS         & 0.16(0.02) & 0.97(0.04) & 0.95(0.04) & 126.73(18.50)  & 22.05(5.08)   & 0.76(0.04) & 0.11(0.01) & 335.57(14.61)  \\ \cline{2-10}
& ElasticNet       & 0.08(0.01) & 0.19(0.01) & 1.00(0.00) & 246.47(10.84)  & 11.68(0.91)   & 1.00(0.00) & 0.88(0.01) & 21.81(1.11)    \\ \cline{2-10}
& Lasso        & 0.07(0.01) & 0.19(0.01) & 1.00(0.00) & 202.05(8.87)   & 8.16(0.60)    & 1.00(0.00) & 0.87(0.01) & 19.88(0.55)    \\ \hline
\end{tabular}%
}
\end{table}

\section*{Appendix B: Proofs}

\begin{proof}[Proof of Proposition~\ref{prop delta}]
Recall that $\hat \Delta = \{\hat \Delta_{kj}\}_{K \times J}$, where
\begin{equation}\label{eq minimizer}
\hat \Delta_{kj} = \argmin_{\Delta_{kj} \in \mathcal{V}}
\Big(\frac{\left\|Y_{j}-P_{X_{k}} Y_{j}\right\|_{F}^{2}}{n-p_k-1} \Delta_{k j}
+\gamma \frac{\left\|P_{X_{k}} Y_{j}\right\|_{F}^{2}}{p_k - 1}\left(1-\Delta_{k j}\right)
\Big).
\end{equation}
The uniqueness of the solution to the above 0-1 integer optimization depends on the uniqueness of $P_{X_k}$, which is divided into two parts. When $p_k \leqslant n $, we have
\[P_{X_k} = X_k(X_k^\t X_k)^{-1}X_k^\t,\]
and the uniqueness of the above equation is certain. When  $p_k > n $,
\[P_{X_{\hat S_{kj}}} = X_{\hat S_{kj}}(X_{\hat S_{kj}}^\t X_{\hat S_{kj}})^{-1}X_{\hat S_{kj}}^\t,\]
and the uniqueness is based on the active set $\hat S_{kj}$. In this paper, we obtain the active set using the Lasso. Based on its properties \citep{tibshirani2012lasso}, the uniqueness of $P_{X_{\hat S_{kj}}}$ is proven. We can equivalently define $\Delta$, in terms of $\Delta_{kj}$, by the following optimization problem:
\[ \hat \Delta =\underset{\Delta \in \mathcal{V}^{K \times J}}{\operatorname{argmin}}~Q(\Delta), \]
where
\[Q( \Delta )= \frac{1}{KJ} \sum_{j=1}^{J} \sum_{k=1}^{K}
\Big(  \frac{1}{n-p_k-1} \left\|Y_{j}-P_{X_{k}} Y_{j}\right\|_{F}^{2} \Delta_{k j}
+\gamma \frac{1}{p_k - 1}\left\|P_{X_{k}} Y_{j}\right\|_{F}^{2}\left(1-\Delta_{k j}\right)
\Big).\]
The uniqueness of $\Delta$ is directly obtained when all the minimizers of \eqref{eq minimizer} are unique.
\end{proof}

\begin{proof}[Proof of Theorem~\ref{thm selection consistency}]
Recall that definitions $J_1 = \{(k, j): \Delta_{kj} = 1\}$ and $J_0 = \{ (k,j): \Delta_{kj} = 0 \}$. Proving $\hat \Delta  = \Delta$ with probability tending to $1$ is equivalent to proving $\hat J_1 = J_1$ and $\hat J_0 = J_0$ with probability tending to $1$. Since \( \hat{J}_0 = \hat{J}_1^c \) and \( J_0 = J_1^c \), it suffices to show \( \hat{J}_1 = J_1 \).

Recall the definition of $l_{kj}$:
\[l_{kj} = \frac{\left\| Y_j - P_{X_k} Y_j \right\|_F^2}{n - p_k - 1} \bigg/ \frac{\left\| P_{X_k} Y_j \right\|_F^2}{p_k - 1}.\]
By Proposition \ref{prop the order of lkj}, for any $(k,j) \in J_1$ and $(k',j') \in J_0$, we have, with probability tending to 1 as $n \to \infty$,
$$l_{kj} = O(n^{3\xi/2-1}) ~\text{and}~ l_{k'j'} = O(n^{\xi/2-\eta}).$$
Since $(\xi/2-\eta) - (3\xi/2-1) = 1 - \xi/2 - \eta > 0$, it follows that as $n \to \infty$,
$$\frac{\min{l_{k'j'}}}{\max{l_{kj}}} \to \infty.$$
We claim the existence of a sequence $\gamma(n)$ such that for sufficiently large $n$:
$$
\max{l_{kj}} \le \gamma(n) \le \min{l_{k'j'}}.
$$
This ensures that $J_1 \subset \hat{J_1}$ and $J_0 \subset \hat{J_0}$, which implies $\hat J_1 = J_1$.
\end{proof}

\begin{proof}[Proof of Proposition~\ref{prop gamma}]

Based on the definition of $\bar{R}_{kj}^2$, we have
\[\bar{R}_{kj}^2 \triangleq 1 - \frac{p_k - 1}{n-p_k-1} \cdot \frac{{\left\|Y_{j}-P_{X_{k}} Y_{j}\right\|_F^{2}}}{{\left\|P_{X_{k}} Y_{j}\right\|_{F}^{2}}},\]
and $l_{kj} = 1 - \bar{R}_{kj}^2$ and $c + \gamma = 1$. When $\Delta_{kj} = 0$, $\bar{R}_{kj}^2$ can be calculated as follows:
\begin{align*}
\bar{R}_{kj}^2
&= 1- \dfrac{\|P_{X_k}^\perp(Y_j - E_j)\|_F^2
+ \langle  2P_{X_k}^\perp(Y_j-E_j), P_{X_k}^\perp E_j \rangle
+ \langle P_{X_k}^\perp E_j, P_{X_k}^\perp E_j \rangle
}{
\|P_{X_{k}}(Y_j - E_j)\|_F^2
+ \langle 2P_{X_k}(Y_j-E_j),P_{X_k}E_j \rangle
+ \langle P_{X_k}E_j,P_{X_k}E_j \rangle
} \cdot \dfrac{p_k - 1}{n - p_k - 1} \\
&= \dfrac{
(n - p_k - 1)(\|P_{X_{k}}(Y_j - E_j)\|_F^2) - (p_k-1)\|P_{X_k}^\perp(Y_j - E_j)\|_F^2
}{
\left [ \|P_{X_{k}}(Y_j - E_j)\|_F^2
+ \langle 2P_{X_k}(Y_j-E_j),P_{X_k}E_j \rangle
+ \langle P_{X_k}E_j,P_{X_k}E_j \rangle) \right ]
\cdot(n - p_k - 1)
} \\
&+
\dfrac{
(n - p_k - 1)\langle 2P_{X_k}(Y_j-E_j),P_{X_k}E_j \rangle
- (p_k-1)\langle  2P_{X_k}^\perp(Y_j-E_j), P_{X_k}^\perp E_j \rangle
}{
\left [ \|P_{X_{k}}(Y_j - E_j)\|_F^2
+ \langle 2P_{X_k}(Y_j-E_j),P_{X_k}E_j \rangle
+ \langle P_{X_k}E_j,P_{X_k}E_j \rangle) \right ]
\cdot(n - p_k - 1)
} \\
&+
\dfrac{
(n - p_k - 1)\langle P_{X_k}E_j,P_{X_k}E_j \rangle
- (p_k-1)\langle P_{X_k}^\perp E_j, P_{X_k}^\perp E_j \rangle
}{
\left [ \|P_{X_{k}}(Y_j - E_j)\|_F^2
+ \langle 2P_{X_k}(Y_j-E_j),P_{X_k}E_j \rangle
+ \langle P_{X_k}E_j,P_{X_k}E_j \rangle) \right ]
\cdot(n - p_k - 1)
}.
\end{align*}
Using Assumption \ref{assumption enhanced symmetry}, the first term on the right-hand side of the equality approximates $0$, while the second and third terms are symmetrically distributed around zero. Thus, we have
\begin{align*}
\mathbb{P}((k,j) \in I(c))
&= \mathbb{P}(\bar{R}_{kj}^2 > c, \Delta = 0) \\
&\approx \mathbb{P}(\bar{R}_{kj}^2 < -c, \Delta = 0) \\
&= \mathbb{P}(1 - l_{kj} < \gamma - 1, \Delta = 0) \\
&=\mathbb{P}(l_{kj} > 2-\gamma) - \mathbb{P}(l_{kj} > 2-\gamma, \Delta = 1)\\
& \approx \mathbb{P}(l_{kj} > 2 - \gamma),
\end{align*}
where the last approximate sign holds since Proposition~\ref{prop the order of lkj} establishes that \(l_{kj} \to 0\) when \(\Delta_{kj} = 1\), the second term vanishes asymptotically. For \(\gamma < 1\), we observe:
\[\mathbb{P}(l_{kj} > 2 - \gamma) \leq \mathbb{P}(l_{kj} > \frac{2 - \gamma}{\gamma}).\]
Rewriting \(l_{kj} = 1 - \bar{R}_{kj}^2\), we have:
\[\mathbb{P}(l_{kj} > \frac{2 - \gamma}{\gamma}) = \mathbb{P}(\bar{R}_{kj}^2 < 1 - \frac{2 - \gamma}{\gamma}).\]
Finally, substituting $c = 1 - \gamma$, we obtain:
\[\mathbb{P}((k,j) \in I(c)) \leq \mathbb{P}(\bar{R}_{kj}^2 < \frac{2c}{c - 1}),\]
which completes the proof.
\end{proof}

\begin{proof}[Proof of Proposition~\ref{prop equiv}]
Set $(\hat \Delta_{11}, \hat B_{\Delta_{11}})$ to be the solution of the following optimization:
\[ (\hat \Delta_{11},\hat B_{\Delta_{11}})= \argmin \Big\{\frac{1}{n_0}\left\|\Delta_{11} \bullet Y_1- X_1B_{\Delta_{11}} \right\|_{F}^{2}
+\gamma W_{11} \left(1-\Delta_{11}\right)   \Big\},
\]
where $n_0 = n-p_1-1$ and $W_{11} = \|  P_{X_1} Y_1 \|_{F}^{2}/(p_1 - 1)$. Since $P_{X_1} = X_1 (X_1^\t X_1)^{-1} X^\t_1$, we have $X_1 \hat B_{\Delta_{11}} = P_{X_1} Y_1 \hat \Delta_{11}$ and then replace the former by the latter in the above function. The above optimization can be transferred as follows
\begin{align*}
& \frac{\left\|\Delta_{11}\bullet Y_1- X_1B_{\Delta_{11}} \right\|_{F}^{2}}{n-p_1-1}
+\gamma W_{11} \left(1-\Delta_{11}\right)  \\
& =  \frac{\left\|\Delta_{11} \bullet Y_1-  P_{X_1} Y_1 \Delta_{11} \right\|_{F}^{2}}{n-p_1-1}
+\gamma \frac{\left\| P_{X_1} Y_1 \right\|_{F}^{2}}{p_1 - 1}\left(1-\Delta_{11}\right)\\
& =  \frac{\left\| Y_1-  P_{X_1} Y_1 \right\|_{F}^{2}}{n-p_1-1} \Delta_{11}
+\gamma \frac{\left\| P_{X_1} Y_1 \right\|_{F}^{2}}{p_1 - 1}\left(1-\Delta_{11}\right)
\end{align*}
Thus, the solutions to the above two functions are equal, and the regression coefficient estimator is unique when $\hat \Delta_{11}$ = 1, i.e.,
$$\hat B_{11} = (X^\t_1 X_1 )^{-1} X_1^\t Y_1.$$
When $\Delta_{11} = 0$, all the elements of coefficient matrix $B$ equal $0$. We allow this situation to exist and expect $\hat \Delta_{11} = 0$ as well as $\hat B = 0$. Thus, the estimator is given as follows:
\[\hat B_{\Delta 11} = (X_1^\t X_1)^{-1} X^\t_1 Y_1 \bullet \hat \Delta_{11}. \]
\end{proof}

\begin{proof}[Proof of Theorem \ref{thm ols}]
Based on Proposition \ref{prop equiv}, the solution to \eqref{eq 11} is unique given by
\[\hat B_{\Delta 11} = (X_1^\t X_1)^{-1} X^\t_1 Y_1 \bullet \hat \Delta_{11}. \]
When we consider the single-block model, to establish that
\[\hat B_{\Delta 11} = (X_1^\t X_1)^{-1} X^\t_1 Y_1 \bullet \Delta_{11}, \]
it is equivalent to proving the block selection consistency, i.e.,
\[ \hat \Delta_{11} = \Delta_{11}, \]
with high probability. For this, we define
\[l_{11} = \dfrac{\left\|Y_{1}-P_{X_{1}} Y_{1}\right\|_{F}^{2}}{n - p_1 - 1} / \dfrac{\left\|P_{X_{1}} Y_{1}\right\|_{F}^{2}}{p_1 - 1},\]
where $P_{X_1}$ is the projection matrix associated with $X_1$. To prove block selection consistency, it suffices to prove the following two parts.

(i) When $\Delta_{11} = 1$, we have $Y_1 = X_1B_{1} + E_1 = T_{11} + E_1$ and obtain:
\begin{align*}
E\|P_{X_1}Y_1\|^2_F = \langle T_{11},T_{11} \rangle
+ E\langle P_{X_1}E_1,P_{X_1}E_1 \rangle
= O(n) + O(n^{\xi}) = O(n),
\end{align*}
where the last equality holds since
\begin{align*}
E\langle P_{X_1}E_1,P_{X_1}E_1 \rangle
\le E\tr(P_{X_1}) \cdot \max_{i} \lambda_i(E_1E_1^\t) = O(n^\xi).
\end{align*}
Thus
$$
E\dfrac{\left\|P_{X_{1}} Y_{1}\right\|_{F}^{2}}{p_1 - 1} = O(n^{1-\xi}).
$$
For the error term, we have
\begin{align*}
E(\|Y_1 - P_{X_1} Y_1 \|_F^2)  \nonumber
= E\langle (I_n-P_{X_1})E_1,(I_n-P_{X_1})E_1 \rangle = O(n-p_1),
\end{align*}
and
$$
E\dfrac{\left\|Y_{1}-P_{X_{1}} Y_{1}\right\|_{F}^{2}}{n -p_1-1} = O(1).
$$
Therefore we have, for $0 < \xi < 1/2$,
\begin{align*}
l_{11}  = \dfrac{\left\|Y_{1}-P_{X_{1}} Y_{1}\right\|_{F}^{2}}{n -p_1-1} / \dfrac{\left\|P_{X_{1}} Y_{1}\right\|_{F}^{2}}{p_1 - 1} = O(n^{\xi - 1}),
\end{align*}
implying that there exists a sequence $\gamma(n)$ such that for large $n$, $l_{11} < kn^{\xi - 1} < \gamma(n)$ for some $k$. 

(ii) When $\Delta_{11} = 0$, we have $Y_1 = E_1$ and
\begin{align*}
E\|X_1 \hat B_{11}\|^2_F &=  E\|P_{X_1}Y_1\|^2_F
=E\langle P_{X_1}E_1,P_{X_1}E_1 \rangle
= O(n^{\xi}).
\end{align*}
Thus
$$
E\dfrac{\left\|P_{X_{1}} Y_{1}\right\|_{F}^{2}}{p_1 - 1} = O(1).
$$
For the error term,
\begin{align*}
E\|Y_1 - X_1 \hat B_{11}\|^2_F = E\langle (I_n-P_{X_1})E_1,(I_n-P_{X_1})E_1 \rangle = O(n-p_1),
\end{align*}
and
$$
E\dfrac{\left\|Y_{1}-P_{X_{1}} Y_{1}\right\|_{F}^{2}}{n -p_1-1} = O(1).
$$
Therefore we have $l_{11} = O(1)$. For the sequence $kn^{\xi - 1} < \gamma(n)$, if $\gamma(n) \to 0$, then for for large $n$, $l_{11} = O(1) > \gamma(n)$.
\end{proof}

\begin{proof}[Proof of Proposition \ref{prop two}]
Recall that we set $(\hat \Delta_{11}, \hat B_{\Delta_{11}})$ to be the solution of the following optimization:
\begin{equation}\label{eq target two}
(\hat \Delta_{11},\hat B_{\Delta_{11}}) =  \argmin_{\Delta_{11} \in \{0,1\}, B_{11}: \beta^{11}_{ll'} = 0,  l \notin \hat S_{11}}  \Big\{\frac{\left\|\Delta_{11} \bullet Y_1- X_1B_{\Delta_{11}} \right\|_{F}^{2}}{n-p_1-1}
+\gamma W_{11} \left(1-\Delta_{11}\right) \Big\},
\end{equation}
where $B_{\Delta_{11}} = B_{11}\bullet\Delta_{11} $ and $W_{11} = \|P_{X_1} Y_1\|^2_F/(p_1-1)$. Based on the definition of $P_{X_1}$ that
\[P_{X_1} =  X_{\hat S_{11}}(X_{\hat S_{11}}^TX_{\hat S_{11}})^{-1}X_{\hat S_{11}}^T, \]
we have $X_1 \hat B_{\Delta_{11}} = P_{X_1} Y_1 \hat \Delta_{11}$ and the following transforming holds:
\begin{align*}
&  \argmin_{\Delta_{11} \in \{0,1\}, B_{11}: \beta^{11}_{ll'} = 0,  l \notin \hat S_{11}}  \frac{\left\|\Delta_{11}\bullet Y_1- X_1B_{\Delta_{11}} \right\|_{F}^{2}}{n-p_1-1}
+\gamma W_{11} \left(1-\Delta_{11}\right)  \\
& =   \argmin_{\Delta_{11} \in \{0,1\}}  \frac{\left\|\Delta_{11} \bullet Y_1-  P_{X_1} Y_1 \Delta_{11} \right\|_{F}^{2}}{n-p_1-1}
+\gamma \frac{\left\| P_{X_1} Y_1 \right\|_{F}^{2}}{p_1 - 1}\left(1-\Delta_{11}\right)\\
& =  \argmin_{\Delta_{11} \in \{0,1\}} \frac{\left\| Y_1-  P_{X_1} Y_1 \right\|_{F}^{2}}{n-p_1-1} \Delta_{11}
+\gamma \frac{\left\| P_{X_1} Y_1 \right\|_{F}^{2}}{p_1 - 1}\left(1-\Delta_{11}\right).
\end{align*}
Then we discuss the uniqueness of the active set $\hat S_{11}$. Since it is obtained from the lasso for the single-block model, based on the properties of the lasso, the uniqueness of $X_1 \hat B_{\text{lasso}}$ implies the uniqueness of $\hat S_{11}$ \citep{tibshirani2012lasso}.
\end{proof}

\begin{proof}[Proof of Theorem \ref{thm lasso}]
Based on Proposition~\ref{prop two}, the estimator is uniquely presented as follows,
\[\hat B_{\Delta 11,\hat S_{11}} = (X_{\hat S_{11}}^TX_{\hat S_{11}})^{-1}X_{\hat S_{11}}^T Y_1 \bullet \hat \Delta_{11}. \]
We aim to prove that the above estimator converges to the oracle solution of the single-block model. The proof is divided into two parts: Part 1. The active set equals the true nonzero set, i.e., $\hat S_{11} = S_{11}$, where $S_{11}$ denotes the set of relevant covariates and $\hat S_{11}$ denotes its estimate, with high probability. Part 2. The block selection consistency under the high-dimensional setting.

Part 1. We first prove the first party by showing that $\hat S_{11} = S_{11}$ with high probability in two cases. The lasso estimator $\hat S_{11}$ satisfies consistency properties similar to those in the regular regression model. By Lemma A.2 in \citep{yyh2020two}, the restricted eigenvalue condition holds in the multi-response setting. Invoking Corollary 2 of \citep{negahban2012unified}, we have, with high probability and some positive constant $M_2$,
\[\|\hat B_{11} - B_{11}\|_F^2 \leq \frac{M_2^2 |S_{11}| \log p_1}{n}.\]
The above results holds naturally considering the properties of the lasso solution under the multi-response model, which can be viewed as the sum of separately solving the lasso under each single-response model. Thus, the above result leads to
\[\|\hat B_{11} - B_{11}\|_\infty \leqslant \sqrt{M^2_2 \cdot |S_{11}| \log p_1 / n} \leqslant M_2 n^{(c_1 + c_2 - 1)/2}. \]
Combining this with the assumption of the minimal nonzero entries of the coefficient matrix,
\[ |B_{\min}| \geqslant M_2 \sqrt{n^{c_1 + c_2-1}}, \]
we have, with probability tending to 1, that
\[\hat S_{11} = S_{11}.\]

Part 2. Following the arguments as the proof of Theorem~\ref{thm ols}, we first prove that when $\Delta_{11} = 1$,
\begin{equation}\label{eq lasso 1}
P(l_{11} \ge \gamma) \rightarrow 0.
\end{equation}
Based on the event that $\{\hat S_{11} = S_{11}\}$ and Assumption \ref{assumption p'},\ref{assumption Ej} and \ref{assumption sparsity}, we have
\begin{align*}
E\|X_1 \hat B_{11}\|^2_F =  E\|P_{X_1}Y_1\|^2_F
&=\langle T_{11},T_{11} \rangle
+ E\langle P_{X_1}E_1,P_{X_1}E_1 \rangle \\
&= \langle X_{S_{11}}B_{S_{11}},(X_{S_{11}}B_{S_{11}})^\t \rangle
+ E\langle P_{X_{S_{11}}}E_1,P_{X_{S_{11}}}E_1 \rangle \\
&= O(n) + O(n^{c_2}) = O(n).
\end{align*}
Similarly,
\begin{align*}
E\|Y_1 - X_1 \hat B_{11}\|^2_F =  E\|Y_1 - P_{X_1}Y_1\|^2_F
&=E\langle (I_n - P_{X_1})E_1,(I_n - P_{X_1})E_1 \rangle \\
&= O(n - |S_{11}|)
\end{align*}
Thus,
\begin{align*}
E\dfrac{\left\|P_{X_{1}} Y_{1}\right\|_{F}^{2}}{|S_{11}| - 1} = O(n^{1-c_2}),\\
E\dfrac{\left\|Y_{1}-P_{X_{1}} Y_{1}\right\|_{F}^{2}}{n -|S_{11}|-1} = O(1),
\end{align*}
and
\begin{align*}
l_{11} = \frac{{\|Y_{1}-P_{X_1} Y_1\|_F^{2}}}{{\|P_{X_1} Y_1\|_{F}^{2}}} \cdot
\frac{|S_{11}|-1}{n-|S_{11}|-1} \approx O(n^{c_2-1}) \to 0,
\end{align*}
implying that there exists a sequence $\gamma(n)$ such that for large $n$, $l_{11} < kn^{c_2 - 1} < \gamma(n)$ for some $k$, leading to $P(l_{11} \ge \gamma(n)) \rightarrow 0$.

When $\Delta_{11} = 0$, we follow the above argument and the argument of the proof of Theorem~\ref{thm ols} and have
\begin{align*}
E\|P_{X_1}Y_1\|^2_F &= O(n^{c_2}) \\
E\|Y_1 - P_{X_1}Y_1\|^2_F &= O(n).
\end{align*}
Thus
$$l_{11} = \frac{{\|Y_{1}-P_{X_1} Y_1\|_F^{2}}}{{\|P_{X_1} Y_1\|_{F}^{2}}} \cdot
\frac{|S_{11}|-1}{n-|S_{11}|-1} \approx O(1),$$
For the sequence $kn^{\xi - 1} < \gamma(n)$, if $\gamma(n) \to 0$, then for for large $n$, $l_{11} \approx O(1) > \gamma(n)$, leading to $P(l_{11} \le \gamma(n)) \to 0$.

\end{proof}

\begin{proof}[Proof of Theorem~\ref{thm 2}]
The joint estimation for the sparse block model is
\begin{equation}\label{eq solution 1}
\hat B_{\Delta}=\underset{\hat B_{\Delta_{kj}} = 0, (k,j) \in \hat J_0}{\operatorname{argmin}}\left\{\|Y-X B\|_{F}^{2}+\lambda \|B \|_1\right\}.
\end{equation}
The convergence of the solution to \eqref{eq solution 1} is determined by the convergence of the indicator and the lasso optimization, in which the former holds based on Theorem~\ref{thm selection consistency}.

Following the argument of the proof of Theorem \ref{thm lasso} and the proof of Lemma A.2 and Lemma A.3 of \citet{yyh2020two}, by requiring the restricted eigenvalue condition,
\[ v^\t(X^\t X/n) v \geqslant \kappa \|v\|^2_2,\]
we have that, with probability tending to 1 and a positive constant $M_2$
\[ \|\hat B - B \|^2_F \leqslant  \dfrac{M^2_2\cdot |S| \log P}{n}, \]
where $S = \cup S_{kj} $ and $P = \sum^K_{k=1} p_k$.
\end{proof}

\begin{proof}[Proof of Theorem~\ref{thm 3}]
Based on Theorem~\ref{thm 2} and the condition of the minimum absolute value of nonzero entries $B$ that
\[ B_{\min}  \geqslant M_3 \sqrt{|S| \log P /n}, \]
where $B_{\min} = \min\{ |\beta^{kj}_{ll'}| : \beta^{kj}_{ll'} \neq 0, k= 1,\dots,K, j = 1,\dots,J \}$
We have that,
\[P(\hat B \neq B) \rightarrow 0.\]
\end{proof}
\bibliographystyle{apalike}
\bibliography{reference}
\end{document}